%% file: main.tex
\documentclass{article}[11pt]

\usepackage[fencedCode,citations,definitionLists,hashEnumerators,smartEllipses,hybrid]{markdown}
\usepackage[utf8]{inputenc}
\usepackage[T1]{fontenc}
\usepackage{xcolor}
\usepackage[british]{babel}
\usepackage[autostyle]{csquotes}
\usepackage[numbers]{natbib}
\usepackage{mathtools}
\usepackage{amsmath, amsthm,amssymb,amstext,amsbsy} 
\usepackage{tabularx}
\usepackage{thmtools,thm-restate}
\usepackage{fullpage}
\usepackage{enumitem}
\setlist[description]{leftmargin=15pt,labelindent=15pt}

\usepackage[bookmarks=false]{hyperref}  
\definecolor{darkgreen}{rgb}{0,0.5,0}
\hypersetup{linktocpage=true,colorlinks=true,citecolor=red,linkcolor=darkgreen}

\usepackage{setspace}

\newif\ifnotes\notestrue

\ifnotes
 \usepackage{color}
 \definecolor{mygrey}{gray}{0.50}
 \newcommand{\notename}[2]{{\textcolor{red}{\footnotesize{\bf (#1:} {#2}{\bf ) }}}}

\else
 
 \newcommand{\notename}[2]{{}}
 
\fi

\newtheorem{theorem}{Theorem}[section]
\newtheorem{claim}[theorem]{Claim}
\newtheorem{proposition}[theorem]{Proposition}
\newtheorem{lemma}[theorem]{Lemma}

\newcommand{\expref}[2]{{\texorpdfstring{\hyperref[#2]{#1~\ref{#2}}}{#1~\ref{#2}}}} 
\newcommand{\secref}[1]{\expref{Section}{#1}}
\newcommand{\thmref}[1]{\expref{Theorem}{#1}}
\newcommand{\clmref}[1]{\expref{Claim}{#1}}

\newcommand{\lref}[1]{\expref{Lemma}{#1}}

\newcommand{\pref}[1]{\expref{Proposition}{#1}}
\newcommand{\figref}[1]{\expref{Figure}{#1}}

\newcommand{\boldp}{\ensuremath{\mathsf{p}}}

\newcommand{\RZ}{\mathbb{Z}}
\newcommand{\R}{\mathbb{R}}
\newcommand{\E}{\mathbb{E}}

\newcommand{\CD}{\mathcal{D}}
\newcommand{\ip}[2]{\langle #1,#2\rangle}

\newcommand{\BR}{\mathbb{R}}
\newcommand{\BE}{\mathbb{E}}
\newcommand{\BP}{\mathbb{P}}

\newcommand{\ind}{\mathbf{1}}

\newcommand{\calP}{\ensuremath{\mathcal{P}}}

\newcommand{\CE}{\mathcal{E}}

\newcommand*\bj{\ensuremath{\boldsymbol{j}}}
\newcommand*\bk{\ensuremath{\boldsymbol{k}}}

\newcommand{\polylog}{\mathrm{polylog}}
\newcommand{\poly}{\mathrm{poly}}

\newcommand{\eps}{\epsilon}
\newcommand{\disc}{\mathsf{disc}}
\newcommand{\sign}{\mathsf{sign}}

\newcommand{\CS}{\mathcal{S}}

\newcommand{\sfp}{\mathsf{p}}

\newcommand{\err}{\mathsf{err}}
\newcommand{\CG}{\mathcal G}
\newcommand{\CT}{T}
\newcommand{\calT}{\mathcal{T}}
\newcommand{\smin}{\eps_\mathsf{min}}
\newcommand{\smax}{\eps_\mathsf{max}}
\newcommand{\cov}{\mathbf{\Sigma}}

\global\long\def\norm#1{\left\Vert #1\right\Vert }

\newcommand{\BS}{\mathbb{S}}
\newcommand{\p}{\mathbb{P}}

\newcommand{\tM}{M^+}
\newcommand{\op}{\mathsf{op}}
\newcommand{\diam}{\mathsf{diam}}
\newcommand{\im}{\mathrm{im}}
\newcommand{\SE}{\mathcal{S}}
\newcommand{\SA}{\mathcal{A}}

\newcommand{\SG}{\mathcal{G}}
\newcommand{\ST}{\mathcal{T}}
\newcommand{\Tr}{\mathrm{Tr}}

\newcommand{\SD}{\mathcal{D}}
\newcommand{\SL}{\mathcal{L}}
\newcommand{\arbnorm}[1]{\left\|#1\right\|_*}
\newcommand{\qlt}{q^l_{j,k}}
\newcommand{\qrt}{q^r_{j,k}}
\newcommand{\dlt}{d^l_{j,k}}
\newcommand{\drt}{d^r_{j,k}}
\newcommand{\SM}{\mathcal{M}}
\newcommand{\SN}{\mathcal{N}}

\title{Online Discrepancy Minimization for Stochastic Arrivals}
\author{Nikhil Bansal\thanks{CWI Amsterdam and TU Eindhoven, \texttt{N.Bansal@cwi.nl}. Supported by the ERC Consolidator Grant 617951 and the NWO VICI grant 639.023.812.} \and
Haotian Jiang\thanks{Paul G. Allen School of CSE, University of Washington, \texttt{jhtdavid@cs.washington.edu}.
Supported in part by the National Science Foundation, Grant Number CCF-1749609, CCF-1740551, DMS-1839116.
} \and
Raghu Meka\thanks{Department of Computer Science, University of California, Los Angeles, \texttt{raghum@cs.ucla.edu}. Supported by NSF Grant CCF-1553605. Part of the work was done while visiting CWI Amsterdam.
} \and
Sahil Singla\thanks{Institute for Advanced Study and Princeton University, \texttt{singla@cs.princeton.edu}. Supported in part by the Schmidt Foundation. Part of the work done while visiting CWI Amsterdam.} \and 
Makrand Sinha\thanks{CWI Amsterdam, \texttt{makrand@cwi.nl}. Supported by the Netherlands Organization for Scientific Research, Grant Number 617.001.351 and the NWO VICI grant 639.023.812.}
}
\date{}

\begin{document}

\maketitle
\pagenumbering{roman}
\input{abstract.tex}

\clearpage
\setcounter{tocdepth}{2}

   \tableofcontents
   

\newpage

\pagenumbering{arabic}
\setcounter{page}{1}

\input{introduction}

\input{overview}

\input{preliminaries}

\input{komlos-tusnady}

\input{banasczycklinear}

\input{multicolor}

\bibliographystyle{alpha}
\bibliography{fullbib}

\end{document}

%% file: abstract.tex

\begin{abstract}
\medskip
In the stochastic online vector balancing problem, vectors $v_1,v_2,\ldots,v_T$ chosen independently from an arbitrary distribution in $\BR^n$ arrive one-by-one and must be immediately given a $\pm$ sign. The goal is to keep the norm of the discrepancy vector, \emph{i.e.}, the signed prefix-sum, as small as possible for a given target norm.
    
    \medskip 
    We consider some of the most well-known problems in discrepancy theory in the above online stochastic setting, and give algorithms that match the known offline bounds up to $\polylog(nT)$ factors. This substantially generalizes and improves upon the previous results of Bansal, Jiang, Singla,  and Sinha (STOC' 20).
    In particular, for the Koml\'{o}s problem where $\|v_t\|_2\leq 1$ for each $t$, our algorithm achieves $\widetilde{O}(1)$ discrepancy with high probability, improving upon the previous $\widetilde{O}(n^{3/2})$ bound. 
    For Tusn\'{a}dy's problem of minimizing the discrepancy of axis-aligned boxes, we obtain an $O(\log^{d+4} T)$ bound for arbitrary distribution over points.
    Previous techniques only worked for product distributions and gave a weaker $O(\log^{2d+1} T)$ bound. 
    We also consider the Banaszczyk setting, where given a symmetric convex body $K$ with Gaussian measure at least $1/2$, our algorithm achieves $\widetilde{O}(1)$ discrepancy with respect to the norm given by $K$ for input distributions with sub-exponential tails.
  
  \medskip
  
  Our results are based on a new potential function approach.
  Previous techniques consider a potential that penalizes large discrepancy, and greedily chooses the next color to minimize the increase in potential.   
  Our key idea is to introduce a potential that also enforces constraints on how the discrepancy vector evolves, allowing us to maintain certain anti-concentration properties. We believe that our techniques to control the evolution of states could find other applications in  stochastic processes and online algorithms. For the Banaszczyk setting, we further enhance this potential by combining it with ideas from generic chaining. Finally, we also extend these results to the setting of online multi-color discrepancy.  
    
    

\end{abstract}

%% file: introduction.tex

\section{Introduction}

We consider the following online vector balancing question, originally
proposed by Spencer \cite{Spencer77}: vectors $v_1,v_2,\ldots,v_T \in \BR^n$
arrive online, and upon the arrival of $v_t$, a sign $\chi_t \in \{\pm 1\}$ must be
chosen irrevocably, so that the $\ell_\infty$-norm of the \emph{discrepancy vector} (signed sum) 
$d_t :=  \chi_1 v_1 + \ldots + \chi_t  v_t$ remains as small as possible. That is, find the
smallest $B$ such that $\max_{t \in [T]} \|d_t\|_\infty \leq B$.
More generally, one can consider the problem of minimizing $\max_{t \in T} \|d_t\|_K$ with respect to arbitrary norms given by a symmetric convex body $K$.

\paragraph{Offline setting.} The offline version of the problem, where the vectors $v_1,\ldots,v_T$ are given in advance, 
has been extensively studied in discrepancy theory, and has various applications \cite{Matousek-Book09,Chazelle-Book01,ChenST-Book14}. Here we study three important problems in this vein: 
\begin{description}
\item[Tusn\'ady's problem.] Given points $x_1,\ldots, x_T \in [0,1]^d$, we want to assign $\pm$ signs to the points, so that for every axis-parallel box, the difference between the number of points inside the box that are assigned a plus sign and those assigned a minus sign is minimized. 
\item[Beck-Fiala and Koml{\'o}s problem.] Given $v_1,\ldots,v_T \in \R^n$ with Euclidean norm at most one, we want to minimize $\max_{t \in T} \|d_t\|_\infty$. After scaling, a special case of the Koml{\'o}s problem is the Beck-Fiala setting where $v_1,\ldots,v_T \in [-1,1]^n$ are $s$-sparse (with at most $s$ non-zeros).
\item[Banaszczyk's problem.] Given $v_1,\ldots,v_T \in \R^n$ with Euclidean norm at most one, and a convex body $K \in \R^n$ with Gaussian measure\footnote{The Gaussian measure $\gamma_n(\SE)$ of a set $\SE \subseteq \R^n$ is defined as $\BP[G \in \SE]$ where $G$ is standard Gaussian in $\BR^n$.} $\gamma_n(K) \geq 1-1/(2T)$, find the smallest $B$ so that there exist signs such that $d_t \in B \cdot K$ for all $t\in [T]$.
\end{description}

One of the most general and powerful results here is due to Banaszczyk~\cite{B12}:  there exist signs such that $d_t \in O(1) \cdot K$ for all $t\in [T]$ for any convex body $K \in \R^n$ with Gaussian measure\footnote{We remark that if one only cares about the final discrepancy $d_T$, the condition in Banaszczyk's result can be improved to $\gamma_n(K)\geq 1/2$ (though, in all applications we are aware of, this makes no difference if $T=\text{poly}(n)$ and makes a difference of at most $\sqrt{\log T})$ for general $T$).}  $\gamma_n(K)\ge 1-1/(2T)$.
In particular, this gives the best known bounds of $O((\log T)^{1/2})$ for the {Koml\'{o}s problem}; for the Beck-Fiala setting, when the vectors are $s$-sparse,  the bound is $O((s \log T)^{1/2})$. 

An extensively studied case, where sparsity plays a key role, is that of {Tusn\'ady's problem} (see~\cite{Matousek-Book09} for a history), where the best known (non-algorithmic) results, building on a long line of work, are an $O(\log^{d-1/2}T)$ upper bound of~\cite{Nikolov-Mathematika19} and an almost matching  $\Omega(\log^{d-1}T)$ lower bound of~\cite{MN-SoCG15}.

In general, several powerful techniques have been developed for offline discrepancy problems over the last several decades, starting with initial non-constructive approaches such as \cite{Beck-Combinatorica81, Spencer85, gluskin-89, Giannopoulos,Banaszczyk-Journal98,B12}, and more recent algorithmic ones such as \cite{Bansal-FOCS10, Lovett-Meka-SICOMP15, Rothvoss14, MNT14,  BansalDG16, LevyRR17, EldanS18, BansalDGL18,DNTT18}. However, none of them applies to the online setting that we consider here.

\paragraph{Online setting.}
A na\"ive algorithm is to pick each sign $\chi_t$ randomly and independently, which by standard tail bounds gives $B = \Theta((T \log n)^{
1/2})$ with high probability. In typical interesting settings, we have $T \ge \poly(n)$, 
and hence a natural question is whether the dependence on $T$ can be
improved from $T^{1/2}$ to say, poly-logarithmic in $T$, and ideally to even match the known offline bounds.

Unfortunately, the $\Omega(T^{1/2})$ dependence is necessary if the adversary is adaptive\footnote{In the sense that the adversary can choose the next vector $v_t$ based on the current discrepancy vector $d_{t-1}$.}: at each time $t$, the adversary can choose the next input vector $v_t$ to be \emph{orthogonal} to $d_{t-1}$, causing $\|d_t\|_2$ to grow as $\Omega(T^{1/2})$ (see \cite{Spencer-Book87} for an even stronger lower bound). Even for very special cases, such as for vectors in $\{-1,1\}^n$, strong $\Omega(2^n)$ lower bounds are known \cite{Barany79}. 
Hence, we focus on a natural \emph{stochastic} model where we relax the power of the adversary and assume that the arriving vectors are chosen in an i.i.d.~manner from some---possibly adversarially chosen---distribution $\boldp$. In this case, one could hope to exploit that $\ip{d_{t-1}}{v_t}$ is not always zero, \emph{e.g.}, due to anti-concentration properties of the input distribution, and  beat the $\Omega(T^{1/2})$ bound.

Recently, Bansal and Spencer~\cite{BansalSpencer-arXiv19}, considered the special case where $\boldp$ is the uniform distribution on all $\{-1,1\}^n$ vectors, and gave an almost optimal $O(n^{1/2} \log T)$ bound for the $\ell_\infty$ norm that holds  with high probability for all $t \in [T]$.
The setting of general distributions $\boldp$ turns out to be harder and was considered recently by \cite{JiangKS-arXiv19} and \cite{BJSS20}, motivated by \emph{envy minimization} problems and an online version of Tusn\'ady's problem. The latter was also considered independently by Dwivedi, Feldheim, Gurel-Gurevich, and Ramadas~\cite{DFGR19} motivated by the problem of placing points uniformly in a grid.


For an arbitrary distribution $\boldp$ supported on vectors in $[-1,1]^n$, \cite{BJSS20} give an algorithm achieving an $O(n^2 \log T)$ bound for the $\ell_\infty$-norm.
In contrast, the best offline bound is $O((n \log T)^{1/2})$, and hence $\widetilde{\Omega}(n^{3/2})$ factor worse, where $\widetilde{\Omega}(\cdot)$ ignores poly-logarithmic factors in $n$ and $T$. 

More significantly, the existing bounds for the online version are much worse than those of the offline version for the case of $s$-sparse vectors (\emph{Beck-Fiala} setting) --- \cite{BJSS20} obtain a much weaker bound of $O(s n \log T)$ for the online setting while the offline bound of $O((s\log T)^{1/2})$ is independent of the ambient dimension $n$. These technical limitations also carry over to the online Tusn\'ady problem, where previous works~\cite{JiangKS-arXiv19,DFGR19,BJSS20} could only handle product distributions. 

To this end, \cite{BJSS20}  propose two key problems in the i.i.d.~setting. First, for a general distribution $\boldp$ on vectors in $[-1,1]^n$, can one get an optimal $\widetilde{O}(n^{1/2})$ or even $\widetilde{O}(n)$ dependence? Second, can one get $\text{poly}(s, \log T)$ bounds when the vectors are $s$-sparse. In particular, as a special case, can one get $(\log T)^{O(d)}$ bounds for the Tusn\'ady problem, when points arrive from an {\em arbitrary} non-product distribution on $[0,1]^d$.

\subsection{Our Results}
In this paper we resolve both the above questions of~\cite{BJSS20}, and prove much more general results that obtain bounds within poly-logarithmic factors of those achievable in the offline setting.

\paragraph{Online Koml{\'o}s and Tusn\'ady settings.}
We first consider Koml{\'o}s' setting for online discrepancy minimization where the vectors have $\ell_2$-norm at most $1$. Recall, the best known offline bound in this setting is $O((\log T)^{1/2})$~\cite{B12}. We achieve the same result, up to poly-logarithmic factors, in the online setting.

\begin{restatable}[Online Koml{\'o}s setting]{theorem}{komlos} \label{thm:komlos}
    Let $\sfp$ be a distribution in $\mathbb{R}^n$ supported on vectors with Euclidean norm at most $1$.
    Then, for vectors $v_1, \ldots, v_T$ sampled i.i.d. from $\sfp$, there is an online algorithm that  with high probability  maintains a discrepancy vector $d_t$ such that $\| d_t \|_\infty = O(\log^4 (nT))$ for all $t \in [T]$. 
\end{restatable}
In particular, for vectors in $[-1,1]^n$ this gives an $O(n^{1/2} \log^4 (nT))$ bound, and for $s$-sparse vectors in $[-1,1]^n$, this gives an $O(s^{1/2} \log^4 (nT))$ bound, both of which are optimal up to poly-logarithmic factors.

The above result implies significant savings for the online Tusn\'ady problem. Call a set $B \subseteq [0,1]^n$ an axis-parallel box if $B = I_1 \times \cdots \times I_n$ for intervals $I_i \subseteq [0,1]$. In the online Tusn\'ady problem, we see points $x_1,\ldots,x_T \in [0,1]^d$ and need to assign signs $\chi_1,\ldots,\chi_T$ in an online manner to minimize the discrepancy of every axis-parallel box at all times. More precisely, for an axis-parallel box $B$, define\footnote{Here, and henceforth, for a set $S$, denote $\ind_S(x)$ the indicator function that is $1$ if $x \in S$ and $0$ otherwise.}
\[ \disc_t(B) := \Big|\chi_1 \ind_B(x_1) + \ldots + \chi_t  \ind_B(x_t)\Big|.\]

Our goal is to assign the signs $\chi_1,\ldots,\chi_t$ so as to minimize $\max_{t \le T} \disc_t(B)$ for every axis-parallel box $B$. 

There is a standard reduction (see \secref{subsec:tusnady}) from the online Tusn\'ady problem to the case of $s$-sparse vectors in $\R^N$ where $s = (\log T)^d$ but the ambient dimension $N$ is $O_d(T^d)$. Using this reduction, along with \thmref{thm:komlos}, directly gives an $O(\log^{3d/2+4} T)$ bound for the online Tusn\'ady's problem that works for any {\em arbitrary} distribution on points, instead of just product distributions as in \cite{BJSS20}. In fact, we prove a more general result where we can choose arbitrary directions to test discrepancy and we use this flexibility (see  \thmref{thm:gen-disc} below) to improve the exponent of the bound further, and essentially match the best offline bound of $O((\log^{d-1/2} T)$ \cite{Nikolov-Mathematika19}. 


\begin{restatable}[Online Tusn\'ady's problem for arbitrary $\boldp$]{theorem}{tusnady}
\label{thm:tusnady}
Let $\boldp$ be an arbitrary distribution on $[0,1]^d$. 
For points $x_1,\ldots,x_T$ sampled i.i.d from $\boldp$, there is an algorithm which selects signs $\chi_t \in \{\pm 1\}$ such that with high probability for every axis-parallel box $B$, we have  $\max_{t \in [T]} \disc_t(B) = O_d(\log^{d+4} T)$.
\end{restatable}

\thmref{thm:komlos} and \thmref{thm:tusnady} follow from the more general result below. 

\begin{restatable}[Discrepancy for Arbitrary Test Directions]{theorem}{test}
\label{thm:gen-disc}
    Let $\SE \subseteq \BR^n$ be a finite set of test vectors with Euclidean norm at most $1$ and $\sfp$ be a distribution in $\BR^n$ supported on vectors with Euclidean norm at most $1$. Then, for vectors $v_1, \ldots, v_T$ sampled i.i.d. from $\sfp$, there is an online algorithm that with high probability maintains a discrepancy vector $d_t$ satisfying
    \begin{align*}
  \max_{z \in \SE} |d_t^\top z| = O((\log (|\SE|) + \log T)\cdot \log^3(nT)) ~\text{ for every } t \in [T].
    \end{align*}
\end{restatable}

In fact, the proof of the above theorem also shows that given any arbitrary distribution on unit test vectors $z$, one can maintain a bound on the exponential moment $\BE_z[\exp(|\ip{d_t}{z}|)]$ at all times. 

The key idea involved in proving \thmref{thm:gen-disc} above, is a novel potential function approach. In addition to controlling the discrepancy $d_t$ in the test directions, we also control how the distribution of $d_t$ relates to the input vector distribution $\boldp$. This leads to better anti-concentration properties, which in turn gives better bounds on discrepancy in the test directions.
We describe this idea in more detail in Sections
\ref{sec:high-level-tech} and \ref{sec:proofOverview}.

\paragraph{Online Banaszczyk setting.}
Next, we consider discrepancy with respect to general norms given by an arbitrary convex body $K$. To recall, in the offline setting, Banaszczyk's seminal result \cite{B12} shows that if $K$ is any convex body with Gaussian measure $1-1/(2T)$, then for any vectors $v_1,\ldots,v_T$ of $\ell_2$-norm at most $1$, there exist signs $\chi_1,\ldots, \chi_T$ such that the discrepancy vectors $d_t \in O(1) \cdot K$ for all $t \in T$. 

Here we study the online version when the input distribution $\sfp \in \BR^n$ has sufficiently good tails. Specifically, we say a univariate random variable $X$ has \emph{sub-exponential tails} if for all $r > 0$, $\BP\big[ |X - \BE[X]| > r \sigma(X)\big]  \leq e^{-\Omega(r)}$, 
where $\sigma(X)$ denotes the standard-deviation of $X$. We say a multi-variate distribution $\sfp \in \BR^n$ has sub-exponential tails if all its one-dimensional projections have sub-exponential tails. That is, $$\BP_{v \sim p}\left[\Big|\ip{v}{\theta} - \mu_\theta] \Big| \ge \sigma_\theta \cdot r\right] \le e^{-\Omega(r)} ~~ \text{ for every } \theta \in \BS^{n-1} \text{ and every } r>0,$$  
where $\mu_\theta$ and $\sigma_\theta$ are the mean and standard deviation\footnote{Note that when the input distribution $\sfp$ is $\alpha$-isotropic, \emph{i.e.} the covariance is $\alpha I_n$, then $\sigma_\theta = \alpha$ for every direction $\theta$, but the above definition is a natural generalization to handle an arbitrary covariance structure.} of the scalar random variable $X_\theta = \ip{v}{\theta}$.

Many natural distributions, such as when $v$ is chosen uniform over the vertices of the $\{\pm 1\}^n$ hypercube (scaled to have Euclidean norm one), uniform from a convex body, Gaussian distribution (scaled to have bounded norm with high probability), or uniform on the unit sphere, have  a sub-exponential tail and in these cases our bounds match the offline bounds up to poly-logarithmic factors.

\begin{restatable}[Online Banaszczyk Setting]{theorem}{Banaszczyk}
\label{thm:gen-disc-ban}
    Let $K \subseteq \BR^n$ be a symmetric convex body with $\gamma_n(K) \ge 1/2$ and $\sfp$ be a distribution with sub-exponential tails that is supported over vectors of Euclidean norm at most 1. Then, for vectors $v_1, \ldots, v_T$ sampled i.i.d. from $\sfp$, there is an online algorithm that  with high probability maintains a discrepancy vector $d_t$ satisfying $d_t \in C \log^5(nT) \cdot K$ for all $t \in [T]$ and a universal constant $C$.
\end{restatable}

The proof of the above theorem, while similar in spirit to \thmref{thm:gen-disc}, is much more delicate. In particular, we cannot use that theorem directly as capturing a general convex body as a polytope may require exponential number of constraints (the set $\SE$ of test vectors). 




\paragraph{Online Weighted Multi-Color Discrepancy.}

Finally we consider the setting of weighted multi-color discrepancy, where we are given vectors $v_1, \ldots, v_T \in \BR^n$ sampled i.i.d. from a distribution $\sfp$  on vectors with $\ell_2$-norm at most one, an integer $R$ which is the number of colors available, positive weights $w_c \in [1,\eta]$ for each color $c \in [R]$, and a norm $\arbnorm{\cdot}$. At each time $t$, the algorithm has to choose a color $c \in [R]$ for the arriving vector, so that the discrepancy $\disc_t$ with respect to $\arbnorm{\cdot}$, defined below, is minimized for every $t \in [T]$:
\begin{align*}
   \disc_t(\arbnorm{\cdot}) := \max_{c\neq c'} \disc_t(c,c') ~\text{ where }~ \disc_t(c,c') := \arbnorm{  \frac{d_c(t) / w_c - d_{c'}(t) / w_{c'} } {1/w_c + 1/w_{c'} }},
\end{align*}
with $d_c(t)$ being the sum of all the vectors that have been given the color $c$ till time $t$. We note that (up to a factor of two) the case of unit weights and $R=2$ is the same as assigning $\pm$ signs to the vectors $(v_i)_{i \le T}$, and we will also refer to this setting as \emph{signed discrepancy}. 

We show that the bounds from the previous results also extend to the setting of multi-color discrepancy.
\begin{restatable}[Weighted multi-color discrepancy]{theorem}{multicolor}\label{thm:multicolor-intro}
 For any input distribution $\sfp$ and any set $\SE$ of $\poly(nT)$ test vectors with Euclidean norm at most one, there is an online algorithm for the weighted multi-color discrepancy problem that maintains discrepancy $O(\log^2(R \eta) \cdot \log^4(nT))$ with the norm $\|\cdot\|_* = \max_{z \in \SE} |\ip{\cdot}{z}|$. 
 
 Further, if the input distribution $\sfp$ has sub-exponential tails then one can maintain multi-color discrepancy $O(\log^2(R \eta)\cdot \log^5(nT))$ for any norm $\|\cdot\|_*$ given by a symmetric convex body $K$ satisfying $\gamma_n(K) \ge 1/2$.
\end{restatable}

As an application, the above theorem implies upper bounds for multi-player envy minimization in the online stochastic setting, as defined in~\cite{BenadeKPP-EC18}, by reductions similar to those in \cite{JiangKS-arXiv19} and \cite{BJSS20}.

We remark that in the offline setting, such a statement with logarithmic dependence in $R$ and $\eta$ is easy to prove by identifying the various colors with leaves of a binary tree and recursively using the offline algorithm for signed discrepancy. It is not clear how to generalize such a strategy to the online stochastic setting, since the algorithm for signed discrepancy might use the stochasticity of the inputs quite strongly. 

By exploiting the idea of working with the Haar basis, we show how to implement such a strategy in the online stochastic setting: we prove that if there is a greedy strategy for the signed discrepancy setting that uses a potential satisfying certain requirements, then it can be converted to the weighted multi-color discrepancy setting in a black-box manner.

\subsection{High-Level Approach} \label{sec:high-level-tech}
Before describing our ideas, it is useful to discuss the bottlenecks in the previous approach. In particular, 
the quantitative bounds for the online Koml\'{o}s problem, as well as for the case of sparse vectors obtained in \cite{BJSS20} are the best possible using their approach, and improving them further required new ideas. We describe these ideas at a high-level here, and refer to Section~\ref{sec:proofOverview}
for a more technical overview.

\paragraph{Limitations of previous approach.} 

For intuition, let us first consider the simpler setting, where we care about minimizing the Euclidean norm of the discrepancy vector $d_t$ --- this will already highlight the main issues. As mentioned before, if the adversary is adaptive in the online setting, then they can always choose the next input vector $v_t$ to be orthogonal to $d_{t-1}$ (i.e., $\ip{d_{t-1}}{v_t} = 0$) causing $\|d_t\|_2$ to grow as $T^{1/2}$. However, if $\ip{d_{t-1}}{v_t}$ is typically large, then one can reduce $\|d_t\|_2$ by choosing $\chi_t = - \sign(\ip{d_{t-1}}{v_t})$, as the following shows:
\begin{equation}\label{eqn:l2}
    \ \|d_t\|_2^2 - \|d_{t-1}\|_2^2 ~=~ 2\chi_t \cdot \ip{d_{t-1}}{v_t} + \|v_t\|^2_2 ~\le~ -2 |\ip{d_{t-1}}{v_t}| + 1. 
\end{equation}

The key idea in \cite{BJSS20}  was that if the vector $v_t$ has uncorrelated coordinates (i.e.~$\E_{v_t \sim \boldp} [v_t(i)v_t(j)]=0$ for $i\neq j$), then one can exploit \emph{anti-concentration} properties to essentially argue that $|\ip{d_{t-1}}{v_t}|$ is typically large when $\|d_{t-1}\|_2$ is somewhat big, and the greedy choice above works, as it gives a \emph{negative drift} for the $\ell_2$-norm. However, uncorrelated vectors satisfy provably weaker anti-concentration properties, by up to a $n^{1/2}$ factor ($s^{1/2}$ for $s$-sparse vectors), compared to those with independent coordinates. 
This leads up to an extra $n^{1/2}$ loss in general. 

Moreover, to ensure uncorrelation one has to work in the eigenbasis of the covariance matrix of $\boldp$, which could destroy sparsity in the input vectors and give bounds that scale polynomially with $n$. \cite{BJSS20} also show that one can combine the above high-level uncorrelation idea with a potential function that tracks a soft version of maximum discrepancy in any coordinate, 
\begin{align}\label{eq:potential}
    \  \Phi_{t-1} = \sum_{i=1}^n \exp( \lambda d_{t-1}(i)),
\end{align}
to even get bounds on the $\ell_\infty$-norm of $d_t$. However, this is also problematic as it might lead to another factor $n$ loss, due to a change of basis (twice). 




To achieve sparsity based bounds in the special case of online Tusn\'ady's problem, previous approaches use the above ideas and exploit the special problem structure. In particular, when the input distribution $\sfp$ is a product distribution, \cite{BJSS20}  (and \cite{DFGR19}) observe that one can work with the natural Haar basis which also has a product structure in $[0,1]^d$ --- this makes the input vectors uncorrelated, while simultaneously preserving the sparsity due to the recursive structure of the Haar basis. However, this severely restricts $\boldp$ to product distributions and previously, it was unclear how to even handle a mixture of two product distributions.


\paragraph{New potential: anti-concentration from exponential moments.} 
Our results are based on a new potential. Typical potential analyses for online problems show that no matter what the current state is, the potential does not rise much when the next input arrives. As discussed above, this is typically exploited in the online discrepancy setting using \emph{anti-concentration} properties of the incoming vector $v_t \sim \sfp$ --- one argues that no matter the current discrepancy vector $d_{t-1}$, the inner product $\ip{d_{t-1}}{v_t}$ is typically large so that a sign can be chosen to decrease the potential (recall \eqref{eqn:l2}).


However, as in \cite{BJSS20}, such a worst-case analysis is restrictive as it requires $\boldp$ to have additional desirable properties such as uncorrelated coordinates. A key conceptual idea in our work is that instead of just controlling a suitable proxy for the norm of the discrepancy vectors $d_t$, we also seek to control structural properties of the distribution $d_t$. Specifically, we also seek to evolve the distribution of $d_t$ so that it  has better anti-concentration properties with respect to the input distribution. In particular, one can get much better anti-concentration for a random variable if one also has control on the higher moments. For instance, if we can bound the fourth moment of the random variable $Y_t \equiv \ip{d_{t-1}}{v_t}$, in terms of its variance, say $\BE[Y_t^4] \ll \BE[Y_t^2]^2$, then the Paley-Zygmund inequality implies that $Y_t$ is far from zero. However, working with $\BE[Y_t^4]$ itself is too weak as an invariant and necessitates looking at even higher moments.

A key idea is that these hurdles can be handled cleanly by looking at another potential that controls the \emph{exponential moment} of $Y_t$. Specifically, all our results are based on an aggregate potential function based on combining a potential of the form \eqref{eq:potential}, which enforces \emph{discrepancy constraints}, together with variants of the following potential, for a suitable parameter $\lambda$, which enforces \emph{anti-concentration constraints}:
$$\Phi_t \sim \BE_v[\exp(\lambda |\ip{d_{t}}{v}|)].$$
 This clearly allows us to control higher moments of $\ip{d_t}{v}$, in turn allowing us to show strong anti-concentration properties without any assumptions on $\boldp$. We believe the above idea of controlling the space of possible states where the algorithm can be present in, could potentially be useful for other applications. 


To illustrate the idea in the concrete setting of $\ell_2$-discrepancy, let us consider the case when the input distribution $\sfp$ is mean-zero and $1/n$-\emph{isotropic}, meaning the covariance $\cov=\BE_{v\sim \sfp}[vv^\top]= I_n/n$. Here, if we knew that the exponential moment $\BE_{v\sim \sfp}[\exp(|\ip{d_{t-1}}{v}|)] \le T$, then  it implies that with high probability $|\ip{d_{t-1}}{v}| \le \log T$ for $v \sim \sfp$. To avoid technicalities, let us assume that $|\ip{d_{t-1}}{v}| \le \log T$ holds with probability one. Therefore, when $v_t$ sampled independently from $\sfp$ arrives, then since $\BE\big[|AB|\big] \ge {\BE[AB]}/{\|B\|_\infty}$ for any coupled random variables $A$ and $B$, taking $A=\ip{d_{t-1}}{v_t}$ and $B = \ip{d_{t-1}}{v_t}/\log T$, we get that 
\[ \BE[|\ip{d_{t-1}}{v_t}|] ~\ge~ \frac{1}{\log T}\cdot \BE_{v_t}[ d_{t-1}^\top v_tv_t^\top d_{t-1}] ~=~  \frac{1}{\log T} \cdot d_{t-1}^\top \cov d_{t-1} ~=~ \frac{\|d_{t-1}\|_2^2}{n\log T}.\]

Therefore, whenever $\|d_{t-1}\|_2 \gg (n\log T)^{1/2}$, then the drift in $\ell_2$-norm of the discrepancy vector $d_t$ is negative. Thus, we can obtain the optimal $\ell_2$-discrepancy bound of $O((n\log T)^{1/2})$.\\





\noindent {\bf Banasaczyk setting.} In the Banaszczyk setting, the algorithm uses a carefully chosen set of test vectors at different scales that come from \emph{generic chaining}. In particular, we use a potential function based on test vectors derived from the generic chaining decomposition of the polar $K^{\circ}$ of the body $K$. 

However, as there can now be exponentially many such test vectors, more care is needed. First, we use that the Gaussian measure of $K$ is large to control the number of test vectors at each scale in the generic chaining decomposition of $K^{\circ}$.
Second, to be able to perform a union bound over the test vectors at each scale, one needs substantially stronger tail bounds than in \thmref{thm:gen-disc}. To do this, we scale the test vectors to be quite large, but this becomes problematic with standard tools for potential analysis, such as Taylor approximation, as the update to each term in the potential can be much larger than potential itself, and hard to control. Nevertheless, we show that if the distribution has sub-exponential tails, then such an approximation holds ``on average'' and the growth in the potential can be bounded. 

\paragraph{Concurrent and Independent Work.}

 In a concurrent and independent work,  Alweiss, Liu, and Sawhney \cite{ALS-arXiv20} obtained online algorithms achieving poly-logarithmic discrepancy bound for the Koml\'os and Tusn\'ady’s problems in the more general setting where the adversary is oblivious. Their techniques, however, are completely different from the potential function based techniques of the present paper.  In fact, as noted by the authors of \cite{ALS-arXiv20}, a potential function analysis encounters significant difficulties here --- the algorithm is required to control the evolution of the discrepancy vectors and such an invariant is difficult to maintain with a potential function, even for stochastic inputs. With the techniques and conceptual ideas we introduce, we can overcome this barrier in the stochastic setting. We believe that our potential-based approach to control the state space of the algorithm could prove useful for other stochastic problems.

%% file: overview.tex
\newcommand{\SV}{\mathscr{V}}

\section{Proof Overview} \label{sec:proofOverview}

Recall the setting: the input vectors $(v_\tau)_{\tau\le T}$ are sampled i.i.d. from $\sfp$ and satisfy $\|v\|_2 \le 1$, and we need to assign signs $\chi_1,\ldots,\chi_T$ in an online manner so as to minimize some target norm of the discrepancy vectors $d_t = \sum_{\tau \leq t} \chi_\tau v_\tau$. Moreover, we may also assume, without loss of generality that the distribution is mean-zero as the algorithm can toss a coin and work with either $v$ or $-v$. This means that the covariance matrix $\cov = \BE_v[vv^\top]$ satisfying $0 \preccurlyeq \cov \preccurlyeq I_n$.



\subsection{Komlos Setting} Here our goal is to minimize $\|d_t\|_\infty$. First, consider the  potential function $\BE_{v\sim \sfp}[\cosh(\lambda ~{d_t^\top v})]$ where $\cosh(a)=\frac12\cdot({e^a+e^{-a}})$. This however only puts anti-concentration constraints on the discrepancy vector and does not track the discrepancy in the coordinate directions. It is natural to add a potential term to enforce discrepancy constraints. In particular, let $\sfp_x = \frac12 \sfp + \frac12 \sfp_y$, where $\sfp_y$ is uniform over the standard basis vectors $(e_i)_{i \le n}$, then the potential   \begin{align}
    \ \Phi_t = \BE_{x\sim \sfp_x}[\cosh(\lambda~ {d_t^\top x})],
\end{align}
allows us to control the exponential moments of $\ip{d_{t-1}}{v_t}$ as well as the discrepancy in the target test directions. In particular, if the above potential $\Phi_t \le \poly(T)$, then we get a bound of $O(\lambda^{-1}\log T)$ on $\|d_t\|_{\infty}$. Next we sketch a proof that for the greedy strategy using the above potential, one can take $\lambda = 1/\log T$, so that the potential remains bounded by $\poly(T)$ at all times.

\begin{claim}[Informal: Bounded Drift] If $\Phi_{t-1} \le T^2$, then $\BE_{v_t}[\Delta\Phi_t] := \BE_{v_t}[\Phi_t - \Phi_{t-1}] \le 2$.
\end{claim}

The above implies using standard martingale arguments, that the potential remain bounded by $T^2$ with high probability and hence $\|d_t\|_\infty = \polylog(T)$ at all times $t \in [T]$.

Let us first make a simplifying assumption that $\cov = I_n/n$ and that at time $t$, the condition $ \lambda |{d^\top_{t-1}} {v_t}| \le 2\log T$ holds with probability $1$. We give an almost complete proof below under these conditions. The first condition can be dealt with by an appropriate decomposition of the covariance matrix as sketched below. The second condition only holds with high probability ($1-1/\poly(T)$), because we have a bound on the exponential moment, but the error event can be handled straightforwardly. 

By Taylor expansion, we have that for all $a$, 
\begin{equation}\label{eqn:taylor}
    \cosh(\lambda (a+\delta)) - \cosh(\lambda a) ~~\le~~   \lambda \sinh(\lambda a)\cdot\delta + \lambda^2|\sinh(\lambda a)|\cdot\delta^2 \qquad \text{ for all } |\delta| \le 1,
\end{equation}
where $\sinh(a) = \frac12\cdot({e^{a} - e^{-a}})$ and we used the approximation that $\cosh(a)\approx |\sinh(a)|$. Therefore, since $d_t = d_{t-1} + \chi_t v_t$, by the above inequality we have 
\begin{align*}
    \ \Delta \Phi_{t} ~~\le~~ \chi_t \cdot \lambda  \BE_x\left[\sinh(\lambda d_{t-1}^\top x)\cdot x^\top v_t\right] +  \lambda^2\BE_x\left[|\sinh(\lambda d_{t-1}^\top x)|\cdot |x^\top v_t|^2\right] ~~:=~~  \chi_t \lambda L +  \lambda^2Q.
\end{align*}

Since the algorithm chooses $\chi_t$ to minimize the potential, we have that $\BE_{v_t}[\Delta \Phi_{t}] \le -\lambda \BE_{v_t}[|L|] + \lambda^2 \BE_{v_t}[Q]$.

\paragraph{Upper bounding the quadratic term:} 

Using that $\cov = \BE_{v_t}[v_tv_t^\top]=I_n/n$, we have 
\begin{align*}
    \ \BE_{v_t}[Q] & ~~=~~ \BE_{v_tx}[|\sinh(\lambda d^\top_{t-1}x)|\cdot x^T v_t v_t^\top x] ~~=~~ \BE_{x}[|\sinh(\lambda d^\top_{t-1}x)| \cdot x^T \cov x] \\
    \ &~~=~~ \frac1n \cdot \BE_{x}[|\sinh(\lambda d^\top_{t-1}x)| \cdot \|x\|^2] ~~\le~~ \frac1n \cdot \BE_{x}[|\sinh(\lambda d^\top_{t-1}x)|],
\end{align*}
where the last inequality used that $\|x\|_2 \le 1$. 

\paragraph{Lower bounding the linear term:} For this we use the aforementioned coupling trick: $\BE_{v_t}[|L|] \ge \BE_{v_t}[LY]/\|Y\|_{\infty}$ for any coupled random variable $Y$ \footnote{Here $\|Y\|_\infty$ denotes the largest value of $Y$ in its support.}. Taking $Y=|d^\top_{t-1}v_t|$, we have that $\|Y\|_{\infty} \le \log T$. Therefore,
\begin{align*}
    \BE_{v_t}[|L|] &~~=~~ \BE_{v_t}\Big|\BE_x\left[\sinh(\lambda d^\top_{t-1}x)\cdot x^\top v_t\right]\Big|  ~\ge~ \frac{1}{\log T}\cdot \BE_{v_tx}\left[\sinh(\lambda d^\top_{t-1}x)|\cdot x^\top v_t v^\top d_{t-1}\right] \\
    \ &~~=~~ \frac1{2n\log T} \cdot \BE_{x}[\sinh(\lambda d^\top_{t-1}x) \cdot d^\top_{t-1}x] ~\ge~ \frac{1}{2n \lambda \log T}\cdot \BE_{x}[|\sinh(\lambda d^\top_{t-1}x)|] - 2,
\end{align*}
using that $\sinh(a)a \ge |\sinh(a)|-2$ for all $a \in \BR$.

Therefore, if $\lambda = 1/(2\log T)$, we can bound the drift in the potential 
\[ \BE_{v_t}[\Delta \Phi_{t}] ~~\le~~ -\frac{\lambda}{2n\log T}\cdot \BE_{x}[|\sinh(\lambda d^\top_{t-1}x)|] + \frac{\lambda^2}{n} \cdot \BE_{x}[|\sinh(\lambda d^\top_{t-1}x)|] + 2 ~~\le~~ 2.\]

\paragraph{Non-Isotropic Covariance.} 

To handle the general case when the covariance $\cov$ is not isotropic, let us assume that all the non-zero eigenvalues are of the form $2^{-k}$ for integers $k\ge 0$. One can always rescale the input vectors and any potential set of test vectors, so that the covariance satisfies the above, while the discrepancy is affected only by a constant factor. See Section \ref{sec:dyadiccov} for details.

With the above assumption $\cov = \sum_{k} 2^{-k}\Pi_k$ where $\Pi_k$ is the orthogonal projection on to the subspace with eigenvalues $2^{-k}$. Since, we only get $T$ vectors, we can ignore the eigenvalues smaller than $(nT)^{-4}$ and only need to consider $O(\log (nT))$ different scales. Then, one can work with the following potential which imposes the alignment constraint in each such subspace:
$$\Phi_t = \sum_{k} \BE_{x\sim \sfp_x}[\cosh(\lambda~ {d_t^\top \Pi_k x})].$$
As we have $O(\log (nT))$ pairwise orthogonal subspaces, we can still choose $\lambda=1/\polylog(nT)$ and with some care, the drift can be bounded using the aforementioned ideas. Once the potential is bounded, we can bound $\|d_t\|_\infty$ as before along with triangle inequality. 

\subsection{Banaszczyk Setting} 

Recall that here we are given a convex body $K$ with Gaussian volume at least $1/2$ and our goal is to bound $K$-norm of the discrepancy vector $\|d_t\|_K$. Here, $\|d\|_K$ intuitively is the minimum scaling $\gamma$ of $K$ so that $d \in \gamma K$. To this end, we will use the dual characterization of $K$: Let $K^\circ = \{y: \sup_{x \in K} |\ip{x}{y}| \leq 1\}$, then $\|d\|_K = \sup_{y \in K^\circ} |\ip{d}{y}|$. 

To approach this first note that the arguments from previous section allow us not only to bound $\|d_t\|_\infty$ but also $\max_{z \in \SE} \ip{d_t}{z}$ for an arbitrary set of \emph{test directions} $\SE$ (of norm at most $1$). As long as $|\SE| \leq \poly(nT)$, we can bound $\max_{z \in \SE} \ip{d_t}{z} = \poly(\log(nT))$. 


However, to handle a norm given by an arbitrary convex body $K$, one needs exponentially many test vectors, and the previous ideas are not enough. To design a suitable test distribution for an arbitrary convex body $K$, we use \emph{generic chaining} to bound $\|d_t\|_K = \sup_{z \in K^\circ} \ip{d_t}{z}$ by choosing epsilon-nets\footnote{We remark that one can also work with admissible nets that come from Talagrand's majorizing measures theorem and probably save a logarithmic factor, but for simplicity we work with epsilon-nets at different scales.} of $K^\circ$ at geometrically decreasing scales. Again let us assume that the $\cov=I_n/n$ for simplicity. 

First, assuming Gaussian measure of $K$ is at least $1/2$, it follows that $\diam(K^\circ)=O(1)$ (see \secref{sec:prelims}). So, one can choose the coarsest epsilon-net at $O(1)$-scale while the finest epsilon-net can be taken at scale $\approx 1/\sqrt{n}$ since by adding the standard basis vectors to the test set, one can control $\|d_t\|_2 \le \sqrt{n}$ (ignoring polylog factors) by using the previous ideas in the Koml\"os setting.

\begin{figure}[h!]
   \centering
   {\includegraphics[width=\textwidth]{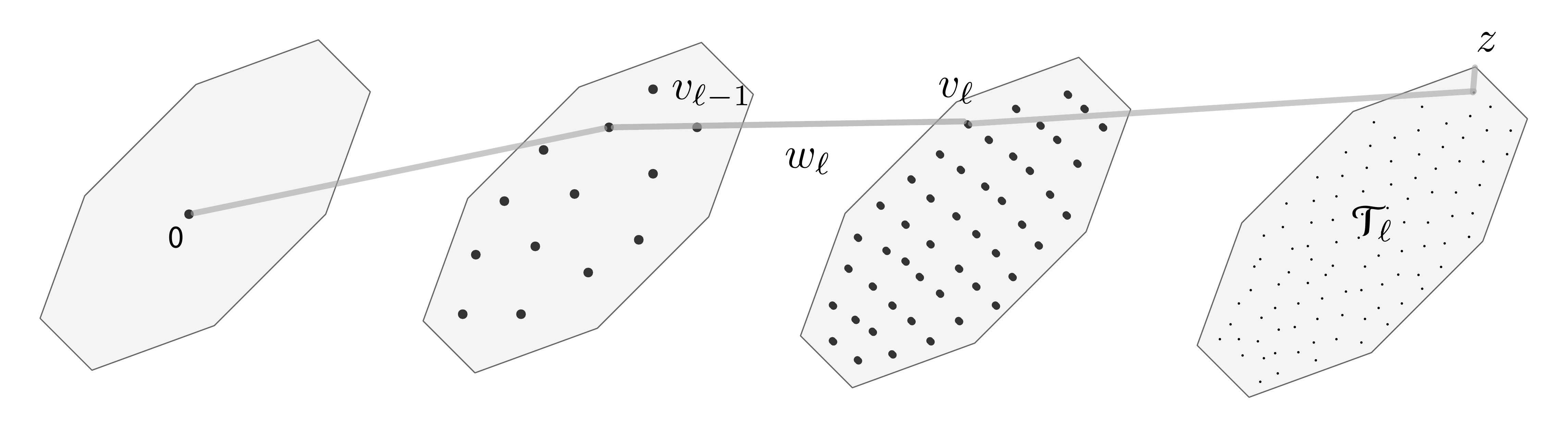}}%
    \caption{\footnotesize The chaining graph $\SG$ showing epsilon-nets of the convex body at various scales. The edges connect near neighbors at two consecutive scales. Note that any point $z \in K^\circ$ can be expressed as the sum of the edge vectors $w_\ell$ where $w_\ell = v_\ell - v_{\ell-1}$, and $(v_{\ell-1}, v_\ell)$ is an edge between two points at scale $2^{-(\ell-1)}$ and $2^{-\ell}$.
    }%
    \label{fig:chaining}%
\end{figure}

Now, one can use generic chaining as follows: define the directed layered graph $\SG$ (see \figref{fig:chaining}) where the vertices $\ST_\ell$ in layer $\ell$ are the elements of an optimal $\eps_\ell$-net of $K^\circ$ with $\eps_\ell=2^{-\ell}$. We add a directed edge from a vertex $u \in \ST_{\ell}$ to vertex $v \in \ST_{\ell+1}$ if $\|u-v\|_2 \le \eps_\ell$ and identify the corresponding edge with the vector $v-u$. The length of any such edge $v-u$, defined as $\|v-u\|_2$, is at most $\eps_\ell$.

Let us denote the set of edges between layer $\ell$ and $\ell+1$ by $\SE_\ell$. Now, one can express any $z \in K^\circ$ as $\sum_\ell w_\ell + w_\err$ where $w_\ell \in \SE_\ell$ and $\|w_\err\|_2 \le 1/\sqrt{n}$. Then, since we can control $\|d_t\|_2 \le \sqrt{n}$, we have
\[ \sup_{z \in K^\circ} \ip{d_t}{z} \le \sum_{\ell} \max_{w \in \SE_\ell} \ip{d}{w} + \max_{\|w\|_2 \le n^{-1/2}} \ip{d}{w_\err} = O(\log n) \cdot \max_\ell \max_{w \in \SE_\ell} \ip{d}{w}.\]
Thus, it suffices to control $\max_{w \in \SE_\ell} \ip{d}{w}$ for each scale using a suitable test distribution in the potential. 

For example, suppose we knew that $\BE_{\widetilde{w}}[\cosh(\lambda d^\top \widetilde{w})] \le T$ for $\widetilde{w}$ uniform in $r^2\cdot\SE_\ell$ for a scaling factor $r^2$. Then, it would follow that $\max_{w \in \SE_\ell} \ip{d}{w} = O(\lambda^{-1}r^{-2} \log |\SE_\ell| \cdot \log T)$.  Standard results in convex geometry (see \secref{sec:prelims}) imply that $|\SE_\ell| \le e^{O(1/\eps_\ell^2)}$, so to obtain a $\polylog(nT)$ bound, one needs to scale the vectors $w \in \SE_\ell$ by a factor of $r = 1/\eps_\ell$. This implies that the $\ell_2$-norm of scaled vector $r^2 \cdot w$ could be as large as $\sqrt{n}$.

This makes the drift analysis for the potential more challenging because now the Taylor expansion in \eqref{eqn:taylor} is not always valid as the update $\delta$ could be as large as $\sqrt{n}$. This is where the sub-exponential tail of the input distribution is useful for us. Since the input distribution is $1/n$-isotropic and sub-exponential tailed,  we know that if $\|w\|_2 \le \sqrt{n}$ , then for a typical choice of $v \sim \sfp$, the following holds
$$\ip{v_t}{w} \approx  \BE_{v_t}[\ip{v_t}{w}^2] = \BE_{v_t}[w^\top vv^\top w] = \frac{\|w\|_2^2}{n} \le 1.$$

Thus, with some work one can show that, the previous Taylor expansion essentially holds "on average" and the drift can be bounded. The case of general covariances can be handled by doing a decomposition as before. Although the full analysis becomes somewhat technical, all the main ideas are presented above.

\subsection{Multi-color Discrepancy}

For the multi-color discrepancy setting, we show that if there is an online algorithm that uses a greedy strategy with respect to a certain kind of potential $\Phi$, then one can adapt the same potential to the multi-color setting in a black-box manner. 

In particular, let the number of colors $R=2^h$ for an integer $h$ and all weights be unit. Let us identify the leaves of a complete binary tree $\calT$ of height $h$ with a color. Our goal is then to assign the incoming vector to one of the leaves. In the offline setting, this is easy to do with a logarithmic dependence of $R$ --- we start at the root and use the algorithm for the signed discrepancy setting to decide to which sub-tree the vector be assigned and then we recurse until the vector is assigned to one of the leaves. Such a strategy in the online stochastic setting is not obvious, as the distribution of the incoming vector might change as one decides which sub-tree it belongs to.

By exploiting the idea used in \cite{BJSS20} and \cite{DFGR19} of working with the Haar basis, we can implement such a strategy if the potential $\Phi$ satisfies certain requirements. Let us define $d_\ell(t)$ to be the sum of all the input vectors assigned to that leaf at time $t$. In the same way, for an internal node $u$ of $\calT$, we can define $d_u(t)$ to be the sum of the vectors $d_\ell(t)$ for all the leaves $\ell$ in the sub-tree rooted at $u$. The crucial insight is then, one can track the difference of the discrepancy vectors of the two children $d^{-}_u(t)$ for every internal node $u$ of the tree $\calT$. In particular, one can work with the potential 
$$\Psi_t = \sum_{u \in \calT} \Phi\big(\beta\:d^-_u(t)\big),$$ for some parameter $\beta$, and assign the incoming vector to the leaf that minimizes the increase in $\Psi_t$. Then, essentially we show that the analysis for the potential $\Phi$ translates to the setting of the potential $\Psi_t$ if $\Phi$ satisfies certain requirements (see \secref{sec:multicolor}).

%% file: preliminaries.tex
\section{Preliminaries}

\subsection{Notation}

 Throughout this paper, $\log$ denotes the natural logarithm unless the base is explicitly mentioned. We use $[k]$ to denote the set $\{1,2,\dotsc, k\}$. Sets will be denoted by script letters (e.g. $\ST$).
 
Random variables are denoted by capital letters (e.g.\ $A$) and values they attain are denoted by lower-case letters possibly with subscripts and superscripts (e.g.\ $a,a_1,a'$, etc.). Events in a probability space will be denoted by calligraphic letters (e.g.\ $\CE$). We also use $\ind_\CE$ to denote the indicator random variable for the event $\CE$. We write $\lambda \sfp + (1-\lambda) \sfp'$ to denote the convex combination of the two distributions.

Given a distribution $\sfp$, we use the notation $x \sim \sfp$ to denote an element $x$ sampled from the distribution $\sfp$. For a real function $f$, we will write $\BE_{x \sim \sfp}[f(x)]$ to denote the expected value of $f(x)$ under $x$ sampled from $\sfp$. If the distribution is clear from the context, then we will abbreviate the above as $\BE_{x}[f(x)]$.

 For a symmetric matrix $M$, we use $\tM$ to denote the Moore-Penrose pseudo-inverse, $\|M\|_{\op}$ for the operator norm of $M$ and $\Tr(M)$ for the trace of $M$. 
 
 \subsection{Sub-exponential Tails}

Recall that a subexponential distribution $\sfp$ on $\BR$ satisfies the following for every $r>0$, $\BP_{x\sim \sfp}[|x - \mu| \ge \sigma r] \le e^{-\Omega(r)}$ where $\mu=\BE_x[x]$ and $\sigma^2=\BE_x[(x-\mu)^2]$. A standard property of a distribution with a sub-exponential tail is \emph{hypercontractivity} and a bound on the exponential moment (c.f. \S2.7 in~\cite{V18}). 

\begin{proposition}
\label{prop:logconcave}
Let $\sfp$ be a distribution on $\BR$ that has a sub-exponential tail with mean zero and variance $\sigma^2$. Then, for a constant $C>0$, we have that $\BE_{x \sim \sfp}[e^{s |x|}] \le C$ for all $|s| \le 1/2\sigma$. Moreover, for every $k>0$, we have $\BE_{x \sim \sfp}[|x|^k]^{1/k} \le C \cdot k \sigma$.
\end{proposition}

\subsection{Convex Geometry}
\label{sec:prelims}
Given a convex body $K \subseteq \BR^n$, its \emph{polar} convex body is defined as $K^\circ = \{ y \mid \sup_{x \in K} |\ip{x}{y}| \le 1\}$. If $K$ is symmetric, then it defines a norm $\|\cdot\|_K$ which is defined as $\|\cdot\|_K = \sup_{y \in K^\circ} \ip{\cdot}{y}$.

For a linear subspace $H \subseteq \BR^n$, we have that $(K \cap H)^\circ = \Pi_H(K^\circ)$ where $\Pi_H$ is the orthogonal projection on to the subspace $H$. 

\paragraph{Gaussian Measure.} We denote by $\gamma_n$ the $n$-dimensional standard Gaussian measure on $\BR^n$. More precisely,  for any measurable set  $\SA\subseteq \BR^n$, we have
\[ \gamma_n(\SA) = \frac{1}{(\sqrt{2\pi})^n} \int_\SA e^{-\|x\|_2^2/2} dx .\]

For a $k$-dimensional linear subspace $H$ of $\BR^n$ and a set $\SA \subseteq H$, we denote by $\gamma_k(\SA)$ the Gaussian measure of the set $\SA$ where $H$ is taken to be the whole space. For convenience, we will sometimes write $\gamma_H(\SA)$ to denote $\gamma_{\dim(H)}(\SA \cap H)$.

The following is a standard inequality for the Gaussian measure of slices of a convex  body. For a proof, see Lemma 14 in \cite{DGLN16}. 
\begin{proposition}\label{prop:slicemeasure}
Let $K \subseteq \BR^n$ with $\gamma_n(K) \ge 1/2$ and $H \subseteq \BR^n$ be a linear subspace of dimension $k$. Then, $\gamma_k(K \cap H) \ge \gamma_n(K)$.
\end{proposition}

\paragraph{Gaussian Width.}
For a set $\ST \subseteq \BR^n$, let $w(\ST) = \BE_g[\sup_{x \in \ST} \ip{g}{x}]$ denote the \emph{Gaussian width} of $\ST$ where $g \in \BR^n$ is sampled from the standard normal distribution. Let $\diam(\ST) = \sup_{x,y \in \ST} \|x-y\|_2$ denote the diameter of the set $\ST$.

The following lemma is standard up to the exact constants. For a proof, see Lemmas 26 and 27 in \cite{DGLN16}.

\begin{proposition}\label{prop:width}
Let $K \subseteq \BR^n$ be a symmetric convex body with $\gamma_n(K) \ge 1/2$. Then, $w(K^\circ)  \le \frac{3}{2}$ and $\diam(K^\circ) \le 4$.
\end{proposition}

To prevent confusion, we remark that the Gaussian width is $\Theta(\sqrt{n})$ factor larger than the \emph{spherical width} defined as $\BE_\theta[\sup_{x \in \ST} \ip{\theta}{x}]$ for a  randomly chosen  $\theta$ from the unit sphere $\BS^{n-1}$. So the above proposition implies that the spherical width of $K^\circ$  is $O(1/\sqrt{n})$. 

For a linear subspace $H \subseteq \BR^n$ and a subset $\ST \subseteq H$, we will use the notation $w_H(\ST) = \BE_g[\sup_{x \in \ST} \ip{g}{x}]$ to denote the Gaussian width of $\ST$ in the subspace $H$, where $g$ is sampled from the standard normal distribution on the subspace $H$. \pref{prop:slicemeasure} and \pref{prop:width} also imply that $w_H(\ST) \le 3/2$.

\paragraph{Covering Numbers.} For a set $\ST \subseteq \BR^n$, let $N(\ST, \eps)$ denote the size of the smallest $\eps$-net of $\ST$ in the Euclidean metric, \emph{i.e.}, the smallest number of closed Euclidean balls of radius $\eps$ whose union covers $\ST$. Then, we have the following inequality (c.f. \cite{W19}, \S5.5).

\begin{proposition}[Sudakov minoration] \label{prop:sudakov}
For any set $\ST \subseteq \BR^n$ and any $\eps > 0$
\[ w(\ST) \ge \frac{\eps}{2} \sqrt{\log N(\ST,\eps)}, ~\text{ or equivalently, } ~ N(\ST, \eps) \le e^{4w(\ST)^2/\eps^2}.\]
\end{proposition}

Analogously, for a linear subspace $H \subseteq \BR^n$ and a subset $\ST \subseteq H$, we also have $w_H(\ST) \ge \frac\eps2 \sqrt{\log N_H(\ST,\eps)}$, where $N_H(\ST,\eps)$ denote the covering numbering of $\ST$ when $H$ is considered the whole space.

%% file: komlos-tusnady.tex

\section{Reduction to $\kappa$-Dyadic Covariance}\label{sec:dyadiccov}

For all our problems, we may assume without loss of generality that the distribution $\sfp$ has zero mean, i.e. $\BE_{v \sim \sfp}[v] = 0$, since our algorithm can toss an unbiased random coin and work with either $v$ or $-v$. Now the covariance matrix $\cov$ of the input distribution $\sfp$ is given by $\cov = \BE_{v \sim \sfp}[vv^\top]$. Since $\|v\|_2 \le 1$, we have that $0 \preccurlyeq \cov \preccurlyeq I$ and $\Tr(\cov)\le 1$.

However, it will be more convenient for the proof to assume that all the non-zero eigenvalues of the covariance matrix $\cov$ are of the form $2^{-k}$ for an integer $k$. In this section, by slightly rescaling the input distribution and the test vectors, we show that one can assume this without any loss of generality.

Consider the spectral decomposition of $\cov = \sum_{i=1}^n \sigma_i u_iu_i^\top$, where $0 \le \sigma_n \le \ldots \le \sigma_1 \le 1$ and $u_1, \ldots, u_n$ form an orthonormal basis of $\mathbb{R}^n$. Moreover, since we only get $T$ vectors, we can essentially ignore all eigenvalues smaller than, say $(nT)^{-8}$, as this error will not affect the discrepancy too much.

For a positive integer $\kappa$ denoting the number of different scales, we say that $\cov$ is $\kappa$-dyadic if every non-zero eigenvalue $\sigma$ is $2^{-k}$ for some $k \in [\kappa]$.
\begin{lemma} \label{lem:covariance_reduction}
Let $\SE \subseteq \BR^n$ be an arbitrary set of test vectors  with Euclidean norm at most $nT$ and $v \sim \sfp$ with covariance $\cov = \sum_i \sigma_i u_iu_i^\top$. Then, there exists a positive-semi definite matrix $M$ with $\|M\|_{\op}\le 1$ such that the covariance of $Mv$ is $\kappa$-dyadic for $\kappa = \lceil8\log (nT)\rceil$. Moreover, there exists a test set $\SE'$ consisting of vectors with Euclidean norm at most  $\max_{y \in \SE} \|y\|$, such that for any signs $(\chi_t)_{t \in T}$, the discrepancy vector $d_t = \sum_{\tau=1}^t \chi_\tau v_\tau$ satisfies 
\[ \max_{y\in \SE} |d_t^\top y| = 2\cdot\max_{z \in \SE'} |(Md_t)^\top z| + O(1).\]
\end{lemma}

\begin{proof}
For notational simplicity, we use $d$ to denote $d_t$. 
We construct matrix $M$ to be postive semi-definite with eigenvectors $u_1, \ldots, u_n$. 
For any $i \in [n]$ such that $\sigma_i \in (2^{-k}, 2^{-k+1}]$ for some $k \in [\kappa]$, we set $M u_i = (2^{k}\sigma_i)^{-1/2} \cdot u_i$, and for every $i \in [n]$ such that $\sigma_i \leq 2^{-\kappa}$, we set $M u_i = 0$. 
It is easy to check that the covariance of $Mv$ for $v \sim \sfp$ is $\kappa$-dyadic. 

We define the new test set to be $\CS' = \{\frac12\tM y\mid y \in \CS\}$ where $\tM$ is the pseudo-inverse of $M$. Note that $\|\tM\|_{\op} \le 2$, so every $z\in \CS'$ satisfies $\|z\|_2 \le \max_{y \in \SE} \|y\| \le  nT$. To upper bound the discrepancy with respect to the test set, let $\Pi_\err$ be the projector onto the span of eigenvectors $u_i$ with $\sigma_i \le 2^{-\kappa}$ and let $\Pi$ be the projector onto its orthogonal subspace. Then, for any $y \in \CS$, we have
\[  |d^\top y| \le |d^\top \Pi y| +  |d^\top \Pi_\err y| \le |(Md)^\top (\tM y)| + nT\cdot \|\Pi_\err d\|_2.\]

By Markov's inequality, with probability at least $1 - (nT)^{-4}$, we have that $\|\Pi_{\err}d\|_2 \le (nT)^{-1}$ and hence, $|d^\top\Pi_{\err}y| = O(1)$ for every $y \in \SE$. It follows that
\[  \max_{y \in \SE}|d^\top y| \le 2\cdot \max_{z\in \SE'}|(Md)^\top z| +  O(1). \qedhere\]
\end{proof}

For all applications in this paper, the test vectors will always have Euclidean norm at most $nT$, so we can always assume without loss of generality that the input distribution $\sfp$, which is supported over vectors with Euclidean norm at most one, has mean $\BE_{v\sim \sfp}[v]=0$, and its covariance $\cov = \BE_v[vv^\top]$ is $\kappa$-dyadic for $\kappa = 8\lceil \log(nT) \rceil$. We will make this assumption in the rest of this paper without stating it explicitly sometimes. 


\section{Discrepancy for Arbitrary Test Vectors} \label{sec:arbitTestVectors}

In this section, we consider discrepancy minimization with respect to an arbitrary set of test vectors with Euclidean length at most $1$. 

\test*


Before getting into the details of the proof, we first give two important applications of \thmref{thm:gen-disc} to the Koml\'os problem in \secref{subsec:komlos} and to the Tusnady's problem in \secref{subsec:tusnady}. 
The proof of \thmref{thm:gen-disc} will be discussed in \secref{subsec:test_theorem_proof}.

\subsection{Discrepancy for Online Koml{\'o}s Setting}
\label{subsec:komlos}

\komlos*


\begin{proof}[Proof of \thmref{thm:komlos}]
Taking the set of test vectors $\SE = \{e_1, \cdots, e_n\}$ where $e_i$'s are the standard basis vectors in $\BR^n$, \thmref{thm:gen-disc} implies an algorithm that w.h.p. maintains a discrepancy vector $d_t$ such that $\|d_t\|_{\infty} = O(\log^4(nT))$ for all $t \in [T]$.
\end{proof}

\subsection{An Application to Online Tusnady's Problem}
\label{subsec:tusnady}

\tusnady*

Firstly, using the probability integral transformation along each dimension, we may assume without loss of generality that the marginal of $\sfp$ along each dimension $i \in [d]$, denoted as $\sfp_i$, is the uniform distribution on $[0,1]$.
More specifically, we replace each incoming point $x \in [0,1]^d$ by $(F_1(x_1), \cdots, F_d(x_d))$, where $F_i$ is the cumulative density function for $\sfp_i$. 
Note that $F_i(x_i)$ is uniform on $[0,1]$ when $x_i \sim \sfp_i$. 
We make such an assumption throughout this subsection. 


A standard approach in tackling Tusn\'ady's problem is to decompose the unit cube $[0,1]^d$ into a canonical set of boxes known as dyadic boxes~(see \cite{Matousek-Book09}). 
Define dyadic intervals $I_{j,k} = [k2^{-j}, (k+1)2^{-j})$ for $j \in \RZ_{\ge 0}$ and $0\le k <2^j$. A dyadic box is one of the form 
\[ B_{\bj,\bk} := I_{\bj(1),\bk(1)} \times \ldots \times I_{\bj(d),\bk(d)},\] 
with $\bj,\bk \in \RZ^d$ such that $0\le \bj$ and $0 \le \bk < 2^{\bj}$, and each side has length at least $1/T$. One can handle the error from the smaller dyadic boxes separately since few points will land in each such box. 
Denoting the set of dyadic boxes as $\SD = \{ B_{\bj,\bk}  \mid  0 \le \bj \le (\log T) \ind ~,~ 0 \le \bk <2^{\bj}\}$, where $\ind \in \BR^d$ is the all ones vector, we note that $|\SD| = O_d(T^d)$. 

Usually, one proves a discrepancy upper bound on the set of dyadic boxes, which implies a discrepancy upper bound on all axis-parallel boxes since each axis-parallel box can be expressed roughly as the disjoint union of $O_d(\log^d T)$ dyadic boxes.
This was precisely the approach used for the online Tusn\'ady's problem in~\cite{BJSS20}. 
However, such an argument has a fundamental barrier. Since each arrival lands in approximately $O_d(\log^d T)$ boxes in $\SD$, 
one can at best obtain a discrepancy upper bound of $O_d(\log^{d/2} T)$ for the set of dyadic boxes, which leads to  $O_d(\log^{3d/2} T)$ discrepancy for all boxes.

Using the idea of test vectors in \thmref{thm:gen-disc}, we can save a factor of $O_d(\log^{d/2} T)$ over the approach above. 
Roughly, this saving comes from the discrepancy of dyadic boxes accumulates in an $\ell_2$ manner as opposed to directly adding up. 
A similar idea was previously exploited by~\cite{BansalG17} for the offline Tusn\'ady's problem.\\

\vspace{5pt}

\begin{proof}[Proof of \thmref{thm:tusnady}]
We view Tusn\'ady's problem as a vector balancing problem in $|\SD|$-dimensions with coordinates indexed by dyadic boxes, where we define $v_t(B) = \ind_B(x_t)$ for each arrival $t \in [T]$ and every dyadic box $B \in \SD$.
Each coordinate $B$ of the discrepancy vector $d_t = \sum_{i=1}^t \chi_i v_i$ is exactly $\disc_t(B)$. 
Notice that $\| v_t \|_2 \leq O_d(\log^{d/2} T)$ since $v_t$ is $O_d(\log^d T)$-sparse. Note that $v_t$'s are the input vectors for the vector balancing problem.

Now we define the set of test vectors $\SE$ that will allow us to bound the discrepancy of any axis-parallel box. For every box $B$ that can be exactly expressed as the disjoint union of several dyadic boxes, i.e. $B = \cup_{B' \in \SD'} B'$ for some subset $\SD' \subseteq \SD$ of disjoint dyadic boxes, we create a test vector $z_B \in \{0,1\}^{|\SD|}$ with $z_B(B') = 1$ if and only if $B' \in \SD'$. 
We call such box $B$ a {\em dyadic-generated} box. Since there are multiple choices of $\SD'$ that give the same dyadic-generated box $B$, we only take $\SD'$ to be the one that contains the smallest number of dyadic boxes. $\SE$ will be the set of all such dyadic-generated boxes.

Recalling that $|\SD|\le 2T$, it follows that $|\SE| = O_d(T^d)$ as each coordinate of a box in $\SE$ corresponds to an endpoint of one of the dyadic intervals in $\CD$. Moreover, every test vector $z_B \in \SE$ is $O_d(\log^d T)$-sparse and thus $\| z_B \|_2 \leq O_d(\log^{d/2} T)$. 
Using \thmref{thm:gen-disc} with both the input and test vectors scaled down by $O_d(\log^{d/2} T)$, we obtain an algorithm that w.h.p. maintains discrepancy vector $d_t$ such that for all $t \in [T]$,
\begin{align*}
    \max_{z_B \in \SE} |d_t^\top z_B| \leq O_d(\log^{d + 4} T) .
\end{align*}

Since $d_t^\top z_B = \disc_t(B)$ which follows from $B$ being a disjoint union of dyadic boxes, we have $\disc_t(B) \leq O_d(\log^{d + 4} T)$ for any dyadic-generated box $B$.

To upper bound the discrepancy of arbitrary axis-parallel boxes, we first introduce the notion of {\em stripes}. 
A stripe in $[0,1]^d$ is an axis-parallel box that is of the form $I_1 \times \cdots \times I_d$ where exactly one of the intervals $I_i$ is allowed to be a proper sub-interval $[a,b] \subseteq [0,1]$. The width of such a stripe is defined to be $b-a$. Stripes whose projection is $[a,b]$ in dimension $i$ satisfying $b - a = 1/T$ correspond to the smallest dyadic interval in dimension $i$. 
We call such stripes {\em minimum dyadic} stripes. 
There are exactly $T$ minimum dyadic stripes for each dimension $i \in [d]$. Since minimum dyadic stripes have width $1/T$ and the marginal of $\sfp$ along any dimension is the uniform distribution over $[0,1]$, a standard application of Chernoff bound implies that w.h.p. the total number of points in all the minimum dyadic stripes is at most $O_d(\log(T))$ points.

For a general axis-parallel box $\widetilde{B}$, it is well-known that $\widetilde{B}$ can be expressed as the disjoint union of a dyadic-generated box $B$ together with at most $k \leq 2d$ boxes $B_1, \ldots, B_{k}$ where  each $B_i \subseteq S_i$ is a subset of a minimum dyadic stripe. We can thus upper bound 
\[
\disc_t(\widetilde{B}) \leq \disc_t(B) + \sum_{i=1}^k \disc_t(B_i) \leq \disc_t(B) + \sum_{i=1}^k r_i.
\]
where $r_i$ is the total number of points in the stripe $S_i$.
As mentioned, w.h.p. we can upper bound $\sum_{i=1}^k r_i = O_d(
\log(T))$ and thus one obtains $\disc_t(\widetilde{B}) = O_d(\log^{d + 4} T)$ for any axis-parallel box $\widetilde{B}$.
This proves the theorem.
\end{proof}

\subsection{Proof of \thmref{thm:gen-disc}}
\label{subsec:test_theorem_proof}

\paragraph{Potential Function and Algorithm.}
By \lref{lem:covariance_reduction}, it is without loss of generality to assume that $\sfp$ is $\kappa$-dyadic, where $\kappa = 8 \lceil \log(nT)\rceil$. For any $k \in [\kappa]$, we use $\Pi_k$ to denote the projection matrix onto the eigenspace of $\cov$ corresponding to the eigenvalue $2^{-k}$ and define $\Pi = \sum_{k=1}^{\kappa} \Pi_k$ to be the sum of these projection matrices. 
Let $\Pi_\err$ be the projection matrix onto the subspace spanned by eigenvectors corresponding to eigenvalues of $\cov$ that are at most $2^{-\kappa}$. 


The algorithm for \thmref{thm:gen-disc} will use a greedy strategy that chooses the next sign so that a certain potential function is minimized. To define the potential, we first define a distribution where some noise is added to the input distribution $\sfp$ to account for the test vectors. 
Let $\sfp_z$ be the uniform distribution over the set of test vectors $\SE$. 
We define the noisy distribution $\sfp_x$ to be $\sfp_x := \sfp/2 + \sfp_z/2$, i.e., a random sample from $\sfp_x$ is drawn  with probability $1/2$ each from  $\sfp$ or $\sfp_z$. Note that any vector $x$ in the support of $\sfp_x$ satisfies $\|x\|_2 \le 1$ since both the input distribution $\sfp$ and the set of test vectors $\SE$ lie inside the unit Euclidean ball. 

At any time step $t$, let $d_{t} = \chi_1 v_1 + \ldots + \chi_t v_t$ denote the current discrepancy vector after the signs $\chi_1, \ldots, \chi_t \in \{\pm1\}$ have been chosen. Set $\lambda^{-1} = 100 {\kappa} \log(nT)$ and define the potential 
\[ \Phi_t ~~=~~ \Phi(d_t) ~~:= ~~ \sum_{k=1}^{\kappa} \BE_{x \sim \sfp_x}\left[\cosh\left(\lambda d_{t}^\top \Pi_k x\right)\right]. \]

When the vector $v_t$ arrives, the algorithm greedily chooses the sign $\chi_t$ that minimizes the increase $\Phi_t - \Phi_{t-1}$. 

\paragraph{Analysis.} 

The above potential is useful because it allows us to give tail bounds on the length of the discrepancy vectors in most directions given by the distribution $\sfp$ while simultaneously controlling the discrepancy in the test directions. In particular, let $\CG_t$ denote the set of {\em good} vectors $v$ in the support of $\sfp$ that satisfy $\lambda|d_t^\top \Pi v| \le {\kappa} \cdot \log (4 \Phi_t/\delta)$.  Then, we have the following lemma.\\

\begin{lemma}\label{lemma:tail}
For any $\delta > 0$ and any time $t$, we have
\begin{enumerate}[label=({\alph*})]
    \item $\p_{v \sim \sfp}(v \notin \CG_t) \le \delta$.
    \item $|d_t^\top \Pi_k z| \le \lambda^{-1}\log (4 |\SE| \Phi_t)$ for all $z \in \SE \text{ and } k \in [k]$.
\end{enumerate}
\end{lemma}
\begin{proof}
\vspace*{1ex}
\begin{enumerate}[label=({\alph*})]
\itemsep1em 
    \item Recall that with probability $1/2$ a sample from $\sfp_x$ is drawn from the input distribution $\sfp$. Using this and the fact that $0 \leq \exp(x) \le 2\cosh(x)$ for any $x \in \BR$, we have   
    $    \sum_{k\in [\kappa]} \BE_{v \sim \sfp}\left[\exp(\lambda |d_t^\top\Pi_k v|)\right] \le 4 \Phi_t$. 
    Note that for any $v \notin \CG_t$, we have $\lambda|d_t^\top \Pi v| \le {\kappa} \cdot \log (4 \Phi_t/\delta)$ by definition, so it follows that $\lambda|d_t^\top \Pi_k v| > \log(4 \Phi_t / \delta)$ for at least one $k \in [\kappa]$. 
    Thus, applying Markov's inequality we get that $\p_{v \sim \sfp}(v \notin \CG_t) \le \delta$.

\item Similarly, a random sample from $\sfp_x$ is drawn from the  uniform distribution over $\SE$ with probability $1/2$, so $\exp\left(\lambda |d^\top \Pi_k z|\right) \le 4 |\SE| \Phi_t$
for every $z \in \SE$ and $k \in [\kappa]$. 
This implies that $|d^\top\Pi_k z| \le  \lambda^{-1} \log (4 |\SE| \Phi_t)$. \qedhere
\end{enumerate}
\end{proof}
\vspace*{8pt}
The next lemma shows that the expected increase in the potential is small on average.

\begin{lemma}[Bounded positive  drift]\label{lemma:drift-komlos} At any time step $t \in [T]$, if $\Phi_{t-1} \leq 3T^5$, then $\BE_{v_t}[\Phi_t] - \Phi_{t-1} \le 2$. 
\end{lemma}

Using \lref{lemma:drift-komlos}, we first finish the proof of \thmref{thm:gen-disc}. 

\begin{proof}[Proof of \thmref{thm:gen-disc}]
We first use \lref{lemma:drift-komlos} to prove that with probability at least $1-T^{-4}$, the potential $\Phi_t \le 3T^5$ for every $t \in [T]$.
Such an argument is standard and has previously appeared  in~\cite{JiangKS-arXiv19,BJSS20}.
In particular, we consider a truncated random process $\widetilde{\Phi}_t$ which is the same as $\Phi_t$ until $\Phi_{t_0} > 3T^5$ for some time step $t_0$; for any $t$ from time $t_0$ to $T$, we define $\widetilde{\Phi}_t = 3T^5$. It follows that $\p[\widetilde{\Phi}_t \geq 3T^5] = \p[\Phi_t \geq 3T^5]$. 
\lref{lemma:drift-komlos} implies that for any time $t \in [T]$, the expected value of the truncated process $\widetilde{\Phi}_t$ over the input sequence $v_1, \ldots, v_T$ is at most $3T$. By Markov's inequality, with probability at least $1-T^{-4}$, the potential $\Phi_t \le 3T^5$ for every $t \in [T]$.


When the potential $\Phi_t \le 3T^5$, part (b) of \lref{lemma:tail} implies that $|d^\top \Pi_k z| = O(\lambda^{-1} \cdot (\log(|\SE|) + \log T))$ for any $z \in \SE$ and $k \in [\kappa]$. Thus, it follows that for every $z \in \SE$,
\[ |d^\top z| ~\le~~ {\sum_{k \in [\kappa]} |d^\top \Pi_k z|} = O({\kappa}\lambda^{-1}(\log(|\SE|) + \log T)) = O((\log(|\SE|) + \log T)\cdot \log^3(nT)),
\]
which completes the proof of the theorem. 
  
\end{proof}

To finish the proof, we prove the remaining \lref{lemma:drift-komlos} next.

\begin{proof}[Proof of \lref{lemma:drift-komlos}]

Let us fix a time $t$. To simplify the notation, let $\Phi = \Phi_{t-1}$ and $\Delta\Phi = \Phi_t - \Phi$, and let $d = d_{t-1}$  and $v = v_t$. 
To bound the change $\Delta \Phi$, we use Taylor expansion. Since $\cosh'(a) = \sinh(a)$ and $\sinh'(a) = \cosh(a)$, for any $a, b \in \BR$ satisfying $|a-b| \le 1$, we have
\begin{align*}
    \  \cosh(\lambda a) - \cosh(\lambda b) &=  \lambda \sinh(\lambda b) \cdot (a-b)  + \frac{\lambda^2}{2!} \cosh(\lambda b) \cdot (a-b)^2 + \frac{\lambda^3}{3!} \sinh(\lambda b)\cdot (a-b)^3 + \cdots , \\[0.8ex]
    \                & \le  \lambda \sinh(\lambda b)  \cdot(a-b) + \lambda^2 \cosh(\lambda b)  \cdot(a-b)^2,\\[1.1ex]
    \   & \le  \lambda \sinh(\lambda b)  \cdot(a-b) + \lambda^2 |\sinh(\lambda b)|  \cdot(a-b)^2 + \lambda^2(a-b)^2,
\end{align*}
where the first inequality follows since $|\sinh(a)| \le \cosh(a)$ for all $a \in \BR$, and since $|a-b|\le 1$ and $\lambda < 1$, so the higher order terms in the Taylor expansion are dominated by the first and second order terms. The second inequality uses that $\cosh(a) \le |\sinh(a)|+1$ for $a \in \BR$.

After choosing the sign $\chi_t$, the discrepancy vector $d_t = d + \chi_t v$. Defining $s_{k}(x) = \sinh(\lambda \cdot d^\top\Pi_k x)$ and noting that $|v^\top \Pi_k x| \le 1$, the above upper bound on the Taylor expansion gives us that 
\begin{align*}
    \ \Delta\Phi &= \sum_{k\in [\kappa]} \BE_{x}\left[\cosh\left(\lambda (d + \chi_t v)^\top \Pi_k x\right)\right] -  \sum_{k\in [\kappa]} \BE_{x}\left[\cosh\left(\lambda d^\top \Pi_k x\right)\right] \\
    &\le  \underbrace{ \chi_t \left (\sum_{k\in [\kappa]} \lambda ~\BE_{x}\left[ s_{k}(x)  v^\top\Pi_k x\right]\right)}_{:=~\chi_t L} + \underbrace{\sum_{k\in [\kappa]} \lambda^2 ~\BE_{x}\left[|s_{k}(x)| \cdot  x^\top\Pi_kvv^\top\Pi_kx\right]}_{:=~Q} + \underbrace{\sum_{k\in [\kappa]} \lambda^2~\BE_{x}\left[ x^\top\Pi_kvv^\top\Pi_kx\right]}_{:=~Q_*},
\end{align*}
where $\chi_t L, Q$, and $Q_*$ denote the first, second, and third terms respectively.
Recall that our algorithm uses the greedy strategy by choosing $\chi_t$ to be the sign that minimizes the potential. 
Taking expectation over the random incoming vector $v \sim \sfp$, we get
\begin{align*}
    \ \BE_{v}[\Delta\Phi] &\le -\BE_{v}[|L|] + \BE_{v}[Q] + \BE_{v}[Q_{*}]. 
\end{align*}

We will prove the following upper bounds on the quadratic (in $\lambda$) terms $Q$ and $Q_*$.
\begin{claim}\label{claim:quadratic}
$ \BE_{v}[Q] \le 2\lambda^2 \sum_{k\in [\kappa]} 2^{-k} ~\BE_{x}[|s_{k}(x)|]$ and $\BE_{v}[Q_*] \le 4\lambda^2.$
\end{claim}

On the other hand, we will show that the linear (in $\lambda$) term $L$ is also large in expectation.
\begin{claim}\label{claim:linear}
$ \BE_{v}[|L|] \ge \lambda B^{-1} \sum_{k\in [\kappa]} 2^{-k} ~ \BE_{x}[|s_{k}(x)|] - 1$ for some value $B \leq 2 {\kappa} \cdot \log(\Phi^2 \kappa n)$. 
\end{claim}

By our assumption that $\Phi \leq 3T^5$, we have that $2\lambda \le B^{-1}$. Therefore, combining the above two claims, we get that 
\[\BE_{v}[\Delta \Phi] \le (2\lambda^2-\lambda B^{-1}) \left(\sum_{k\in [\kappa]} 2^{-k} ~ \BE_{x}[|s_{k}(x)|]\right) + 1 + 4\lambda^2 \le 2.\]

This finishes the proof of  \lref{lemma:drift-komlos} assuming the claims which we prove next.
\end{proof}

\vspace*{10pt}

\begin{proof}[Proof of \clmref{claim:quadratic}]
Recall that $\BE_{v}[vv^\top] = \cov$ and that $\Pi_k \cov \Pi_k = 2^{-k} \Pi_k $. 
Using linearity of expectation,
\begin{align*}
    \ \BE_{v}[Q] ~~=~~ \sum_{k\in [\kappa]} \lambda^2 ~ \BE_{x}[|s_{k(x)}| \cdot x^\top\Pi_k \cov \Pi_k x] ~~&=~~ \lambda^2 \sum_{k\in [\kappa]} 2^{-k} ~ \BE_{x}[|s_{k}(x)| \cdot x^\top \Pi_k x]\\
    &\le ~~ 2\lambda^2 \sum_{k\in [\kappa]} 2^{-k} ~ \BE_{x}[|s_{k}(x)|],
\end{align*}
where the last inequality uses that $\norm{x}_2 \leq 1$. 
Similarly,  
\[ \BE_{v}[Q_{*}] ~~=~~ \sum_{k\in [\kappa]} \lambda^2~\BE_{x}\left[ x^\top\Pi_k\cov\Pi_kx\right]  ~~\le ~~ 2\lambda^2 \sum_{k\in [\kappa]} 2^{-k} \le 4\lambda^2.\qedhere\]
\end{proof} 

\vspace*{10pt}

\begin{proof}[Proof of \clmref{claim:linear}]

To lower bound the linear term, we use the fact that $|L(v)| \ge {\|f\|^{-1}_{\infty}} \cdot f(v) \cdot L(v)$ for any real-valued non-zero function $f$. We will choose the function $f(v) = d^\top\Pi v \cdot \ind_{\CG}(v)$ where $\CG$ will be the event that $|d^\top\Pi v|$ is small, which we know is true because of  \lref{lemma:tail}.\\

In particular, set $\delta^{-1} = \lambda \Phi T$ and let $\CG$ denote the set of vectors $v$ in the support of $\sfp$ such that $\lambda|d^\top \Pi v| \le {\kappa} \cdot \log (4 \Phi/\delta) := B$. 
Then, $f(v) = d^\top\Pi v \cdot \ind_{\CG}(v)$ satisfies $\|f\|_\infty \le \lambda^{-1} B$, and we can lower bound,
\begin{align} \label{eqn:lterm}
     \BE_{v}[|L|] &\ge \frac{\lambda}{\lambda^{-1} B} \sum_{k\in [\kappa]}  \BE_{v,x} [s_{k}(x) \cdot d^\top \Pi v \cdot v^\top \Pi_k x\cdot \ind_{\CG}(v)] \nonumber \\
    &= \frac{\lambda^2}{B} \sum_{k\in [\kappa]} \BE_{x} [s_{k}(x) \cdot d^\top \Pi \cov \Pi_k x] - \frac{\lambda^2}{B} \sum_{k\in [\kappa]} \BE_{x} [s_{k}(x) \cdot d^\top \Pi \cov_\err \Pi_k x], 
\end{align}
where $\cov_\err = \BE_{v}[vv^\top (1 -  \ind_{\CG}(v))]$ satisfies $\|\cov_\err\|_{\op} \le \p_{v \sim \sfp}(v \notin \CG) \le \delta$ using \lref{lemma:tail}.
To bound the first term in \eqref{eqn:lterm}, recall that $s_k(x) = \sinh(\lambda d^\top \Pi_k x)$. Using $\Pi \cov \Pi_k = 2^{-k} \Pi_k$ and the fact that $\sinh(a)a \ge |\sinh(a)| - 2$ for any $a \in \BR$, we have  
\begin{align*}
    \lambda ~\BE_{x} [s_{k}(x) \cdot d^\top \Pi \cov \Pi_k x] ~~=~~ 2^{-k} ~\BE_{x} [s_{k}(x) \cdot \lambda d^\top \Pi_k x] ~~\ge~~ 2^{-k}~ \left(\BE_{x} [|s_{k}(x)|] - 2\right).
\end{align*}

For the second term, we use the bound $\|\cov_\err\|_\op \leq \delta$ to obtain 
\begin{align*}
|d^\top \Pi \cov_\err \Pi_k x | ~~\le ~~ \| \cov_\err \|_{\op} \cdot \| d\|_2 \cdot \|x\|_2 ~~\le ~~ \delta \|d\|_2. 
\end{align*}

Since $\|d\|_2 \le T$ always holds, by our choice of $\delta$, 
\begin{align*}
    \lambda |d^\top \Pi \cov_\err \Pi_k x| \le  \Phi^{-1}.
\end{align*}

Plugging the above bounds in \eqref{eqn:lterm}, 
\begin{align*}
    \BE_{v}[|L|] &\ge  \frac{\lambda}{B} \sum_{k \in [\kappa]} 2^{-k} ~\left(\BE_{x} [|s_{k}(x)|] - 2\right) -  \frac{\lambda}{B} \cdot \Phi^{-1} \left( \sum_{k \in [\kappa]}  \BE_x[|s_k(x)|] \right) \\ 
    &\ge  \frac{\lambda}{B} \sum_{k \in [\kappa]} 2^{-k} ~\BE_{x} [|s_{k}(x)|] - \frac{\lambda}{B} \sum_{k \in [\kappa]} 2^{-k+1} - \frac{\lambda}{B}\\
    \ & \ge \frac{\lambda}{B} \sum_{k \in [\kappa]} 2^{-k} ~\BE_{x} [|s_{k}(x)|] - 1,
\end{align*}
where the second inequality follows since $\sum_{k \in [\kappa]} \BE_x[|s_k(x)|] \le \Phi$.
\end{proof}

%% file: banasczycklinear.tex

\section{Discrepancy with respect to Arbitrary Convex Bodies}

Our main result of this section is the following theorem. 

\Banaszczyk*


\subsection{Potential Function and Algorithm}

As in the previous section, it is without loss of generality to assume that $\sfp$ is $\kappa$-dyadic, where $\kappa = 8 \lceil \log(nT)\rceil$. For any $k \in [\kappa]$, recall that $\Pi_k$ denotes the projection matrix onto the eigenspace of $\cov$ corresponding to the eigenvalue $2^{-k}$ and $\Pi = \sum_{k=1}^{\kappa} \Pi_k$. Further, let us also recall that $\Pi_\err$ is the projection matrix onto the subspace spanned by eigenvectors corresponding to eigenvalues of $\cov$ that are at most $2^{-\kappa}$. We also note that $\dim(\im(\Pi_k)) \le \min\{2^{k},n\}$ since $\Tr(\cov) \le 1$.

Our algorithm to bound the discrepancy with respect to an arbitrary symmetric convex body $K \subseteq \BR^n$ with $\gamma_n(K) \ge 1/2$ will use a greedy strategy with a similar potential function as in \S\ref{sec:arbitTestVectors}. 
Let $\sfp_z$ be a distribution on \emph{test vectors} in $\BR^n$ that will be specified later. Define the noisy distribution $\sfp_x = \sfp/2 + \sfp_z/2,$ \emph{i.e}, a random sample from $\sfp_x$ is drawn from $\sfp$ or $\sfp_z$ with probability $1/2$ each.

At any time step $t$, let $d_{t} = \chi_1 v_1 + \ldots + \chi_t v_t$ denote the current discrepancy vector after the signs $\chi_1, \ldots, \chi_t \in \{\pm1\}$ have been chosen. Set $\lambda^{-1} = 100 {\kappa}\log(nT)$, and define the potential 
\[ \Phi_t = \Phi(d_t) :=  \sum_{k\in [\kappa]} \BE_{x\sim \sfp_x}\left[\exp\left(\lambda ~d_{t}^\top \Pi_k x \right)\right].\]

When the vector $v_t$ arrives, the algorithm chooses the sign $\chi_t$ that minimizes the increase $\Phi_t - \Phi_{t-1}$.

\paragraph{Test Distribution.} To complete the description of the algorithm, we need to choose a suitable distribution $\sfp_z$  on test vectors to give us control on the norm $\|\cdot\|_K = \sup_{y \in K^\circ} \ip{\cdot}{y}$. For this, we will use generic chaining. 

First let us denote by $H_k = \im(\Pi_k)$ the linear subspace that is the image of the projection matrix $\Pi_k$ where the subspaces $\{H_k\}_{k\in [\kappa]}$ are orthogonal and span $\BR^n$. Moreover, recall that $\dim(H_k) \le \min\{2^k,n\}$.

Let us denote by $K_k = K \cap H_k$ the slice of the convex body $K$ with the subspace $H_k$. \pref{prop:slicemeasure} implies that $\gamma_{H_k}(K) \ge 1/2$ for each $k\in [\kappa]$ and combined with \pref{prop:width} this implies that $K^\circ_k := (K_k)^\circ = \Pi_k (K^\circ)$ satisfies $\diam(K^\circ_k) \le 4$ and $w_{H_k}(K^\circ_k) \le 3/2$ for every $k$.

Consider $\eps$-nets of the polar bodies $K^\circ_k$ at geometrically decreasing dyadic scales. Let 
\[ \smin(k) = 2^{-\left\lceil \log_2\left(\frac{1}{10\lambda}\sqrt{\dim(H_k)} \right)\right\rceil} \text{ and } \smax(k) = 2^{-\log_2 \lceil 1/\diam(K_k^\circ) \rceil},\] be the finest and the coarsest scales for a fixed $k$, and for integers $\ell \in [\log_2(1/\smax(k)), \log_2(1/\smin(k))]$, define the scale $\eps(\ell,k) = 2^{-\ell}$. We call these \emph{admissible} scales for any fixed $k$.

Note that for a fixed $k\in [\kappa]$, the number of admissible scales is at most $2\log_2(nT)$ since $\diam(K^\circ_k) \le 4$. The smallest scale is chosen because with high probability we can always control the Euclidean norm of the discrepancy vector in the subspace $H_k$ to be $ \lambda^{-1}\log(nT) \sqrt{\dim(H_k)}$ using a test distribution as used in Komlos's setting. 

Let $\ST(\ell,k)$ be an optimal $\eps(\ell,k)$-net of $K^\circ_k$.
For each $k$, define the following directed layered graph $\SG_k$ (recall \figref{fig:chaining}) where the vertices in  layer $\ell$ are the elements of $\ST(\ell,k)$. Note that the first layer indexed by $\log_2(1/\smax(k))$ consists of a single vertex, the origin. We add a directed edge from $u \in \ST(\ell,k)$ to $v \in \ST(\ell+1,k)$ if $\|v-u\|_2 \le \eps(\ell,k)$. We identify an edge $(u,v)$ with the vector $v-u$ and define its length as $\|v-u\|_2$. Let $\SE(\ell,k)$ denote the set of edges between layer $\ell$ and $\ell+1$. Note that any  edge $(u,v) \in \SE(\ell,k)$ has length at most $\eps(\ell,k)$ and since $w_{H_k}(K^\circ_k) \le 3/2$, \pref{prop:sudakov} implies that, 
\begin{equation}\label{eqn:edges}
    \ |\SE(\ell,k)| ~~\le~~ |\ST({\ell+1},k)|^2 ~~\le~~ 2^{16/\eps(\ell,k)^2}.
\end{equation}

Pick the final test distribution as $\sfp_z = \sfp_\cov/2 + \sfp_y/2$ where $\sfp_\cov$ and $\sfp_y$ denote the distributions given in \figref{fig:test}. 
\begin{figure}[!h]
\begin{tabular}{|l|}
\hline
\begin{minipage}{\textwidth}
\vspace{1ex} 
\begin{enumerate}[label=({\alph*})]
    \item $\sfp_\cov$ is uniform over the eigenvectors $u_1, \ldots, u_n$ of the covariance matrix $\cov$.
    \item $\sfp_y$ samples a random vector as follows: pick an integer $k$ uniformly from $[\kappa]$ and an admissible scale $\eps(\ell,k)$ with probability $\dfrac{2^{-2/\eps(\ell,k)^2}}{\sum_{\ell} 2^{-2/\eps(\ell,k)^2}}$. Choose a uniform vector from $r(\ell,k)^2 \cdot \SE(\ell,k)$, where the scaling factor $r(\ell,k) := 1/\eps(\ell,k)$.
\end{enumerate}
\vspace{0.1ex}
\end{minipage}\\
\hline
\end{tabular}
\caption{Test distributions $\sfp_\cov$ and $\sfp_y$}
\label{fig:test}
\end{figure}

The above test distribution completes the description of the algorithm. Note that adding the eigenvectors will allow us to control the Euclidean length of the discrepancy vectors in the subspaces $H_k$ as they form an orthonormal basis for these subspaces. Also observe that, as opposed to the previous section, the test vectors chosen above may have large Euclidean length as we scaled them. For future reference, we note that the entire probability mass assigned to length $r$ vectors in the support of $\sfp_y$ is at most $2^{-2r^2}$ where $r \ge 1/4$. 

\subsection{Potential Implies Low Discrepancy}

The test distribution $\sfp_z$ is useful because of the following lemma.
In particular, a $\poly(n,T)$ upper bound on the potential function implies a polylogarithmic discrepancy upper bound on $\|d_t\|_K$. 

\begin{lemma} \label{lemma:chaining}
At any time $t$, we have that
\[ \|\Pi_k d_t\|_2 \le \lambda^{-1} \log (4 n \Phi_t)\sqrt{\dim(H_k)} ~~\text{ and }~~ \|d_t\|_K \le O({\kappa} \cdot \lambda^{-1} \cdot  \log(nT) \cdot \log(\Phi_t)).\]
\end{lemma}
\begin{proof}
To derive a bound on the Euclidean length of $\Pi_kd_t$, we note that a random sample from $\sfp_x$ is drawn from the uniform distribution over $\{u_i\}_{i\le n}$ with probability $1/4$, so $\exp\left(\lambda |d_t^\top \Pi_k u_i|\right) \le 4n \Phi_t$
for every $k \in [\kappa]$ and every $i \in [n]$. Since $\{u_i\}_{i \le n}$ also form an eigenbasis for $\Pi$, we get that $|d_t^\top \Pi_k u_i| \le  \lambda^{-1} \log (4 n \Phi_t)$ which implies that $\|\Pi_k d_t\|_2 \le \lambda^{-1} \log (4 n \Phi_t)\sqrt{\dim(H_k)}$.

To see the bound on $\|d_t\|_K$, we note that 
\begin{equation}\label{eqn:chaining}
    \ \|d_t\|_K = \sup_{y \in K^\circ} \ip{d_t}{y} ~\le~ {\sum_{k \in [\kappa]} ~\sup_{y \in K^\circ_k} \ip{\Pi_kd_t}{y}} ~\le~ {\sum_{k \in [\kappa]} \left(\sup_{z \in \ST(\ell,k)} |d_t^\top \Pi_k z| + \smin(k)\|\Pi_kd_t\|_2\right)},\\
\end{equation}
where the last inequality holds since $\ST(\ell,k)$ is an $\smin(k)$-net of $K^\circ_k$. 
By our choice of $\smin(k)$ and the bound on $\|\Pi_k d_t\|_2$ from the first part of the Lemma, we have that $\smin(k)\|\Pi_k d_t\|_2 \le 10 \log(4n \Phi_t)$.\\

To upper bound $\sup_{z \in \ST(\ell,k)} \ip{\Pi_kd_t}{z}$, we pick any arbitrary $z \in \ST(\ell,k)$ and consider any path from the origin to $z$ in the graph $\SG_k$. Let $(u_\ell,u_{\ell+1)}$ be the edges of this path for $\ell \in [\log_2(1/\smin),\log_2(1/\smax)]$ where $u_\ell = 0$ for $\ell=\log_2(1/\smax)$ and $u_\ell=z$ for $\ell = \log_2(1/\smin)$. Then $z = \sum_{\ell} w_\ell$ where $w_\ell = (u_{\ell+1}-u_\ell)$. By our choice of the test distribution, the bound on the potential implies the following for any edge $w \in \SE(\ell,k)$, 
\[\exp\left( \lambda \cdot r(\ell,k)^2 \cdot |d_t^\top\Pi_k w|\right) ~\le~ 2^{2/\eps(\ell,k)^2}\cdot |\SE(\ell,k)| \cdot 4 \Phi_t ~\le~  2^{18/\eps(\ell,k)^2} \cdot 4 \Phi_t,\]
where the second inequality follows from $|\SE(\ell,k)| \le 2^{16/\eps(\ell,k)^2}$ in \eqref{eqn:edges}. This implies that for any edge $w \in \SE(\ell,k)$, 
\[ |d_t^\top \Pi_k w| ~\le~ \lambda^{-1}\log(4 \Phi_t).
\]

Since $z = \sum_{\ell} w_\ell$ and there are at most $\log(n)$ different scales $\ell$, we get that $|d_t^\top \Pi_k z| \le \lambda^{-1} \cdot \log(n) \cdot \log (4 \Phi_t)$. Since $z$ was arbitrary in $\ST(\ell,k)$, plugging the above bound in \eqref{eqn:chaining} completes the proof.
\end{proof}

The next lemma shows that the expected increase (or drift) in the potential is small on average.
\begin{lemma}[Bounded Positive Drift]\label{lemma:gen-drift-ban} Let $\sfp$ be supported on the unit Euclidean ball in $\BR^n$ and has a sub-exponential tail. There exist an absolute constant $C > 0$ such that if $ \Phi_{t-1} \le T^5$ for any $t$, then $\BE_{v_t \sim \sfp}[\Phi_t] - \Phi_{t-1} \le C$.
\end{lemma}

Analogous to the proof of \thmref{thm:gen-disc}, \lref{lemma:gen-drift-ban} implies that w.h.p. the potential $\Phi_t \le T^5$ for every $t \in [T]$. Combined with \lref{lemma:chaining}, and recalling that $\kappa = O(\log nT)$ and $\lambda^{-1}=O({\kappa} \log(nT))$, this proves \thmref{thm:gen-disc-ban}. To finish the proof, we prove \lref{lemma:gen-drift-ban} in the next section.

\subsection{Drift Analysis: Proof of \lref{lemma:gen-drift-ban}} 
The  proof is quite similar to the analysis for Komlos's setting. In particular, we have the following tail bound analogous to \lref{lemma:tail}. 
Let $\CG_t$ denote the set of {\em good} vectors $v$ in the support of $\sfp$ that satisfy $\lambda|d_t^\top \Pi v| \le {\kappa} \cdot \log (4 \Phi_t/\delta)$.

\begin{lemma}\label{lemma:gen-tail-ban}
For any $\delta > 0$ and any time $t$, we have $\BP_{v \sim \sfp}(v \notin \CG_t) \le \delta$.
\end{lemma}
We omit the proof of the above lemma as it is the same as that of \lref{lemma:tail}.

\begin{proof}[Proof of  \lref{lemma:gen-drift-ban}]
Recall that our potential function is defined to be
\[ \Phi_t := \sum_{k\in [\kappa]} \BE_{x\sim \sfp_x}\left[\exp\left(\lambda ~d_{t}^\top \Pi_k x \right)\right],
\]
where $\sfp_x = \sfp/2 + \sfp_{\cov}/4 + \sfp_y / 4$ is a combination of the input distribution $\sfp$ and test distributions $\sfp_{\cov}$ and $\sfp_y$, each constituting a constant mass. 

Let us fix a time $t$. To simplify the notation, we denote $\Phi = \Phi_{t-1}$ and $\Delta\Phi = \Phi_t - \Phi$, and denote $d = d_{t-1}$  and $v = v_t$. 
To bound the potential change $\Delta \Phi$, we use the following inequality, which follows from a modification of the Taylor series expansion of $\cosh(r)$ and holds for any $a,b \in \BR$,
\begin{align}\label{eqn:taylor_exp}
    \ \cosh(\lambda a)-\cosh(\lambda b) & \le \lambda \sinh(\lambda b) \cdot (a - b) + \frac{\lambda^2}{2} \cosh(\lambda b) \cdot e^{|a-b|}(a - b)^2.
\end{align}
Note that when $|a-b| \ll 1$, then $e^{|a-b|} \le 2$, so one gets the first two terms of the Taylor expansion as an upper bound, but here we will also need it when $|a-b|\gg 1$.

Note that every vector in the support of $\sfp$ and $\sfp_{\cov}$ has Euclidean length at most $1$, while $y \sim \sfp_y$ may have large Euclidean length due to the scaling factor of $r(\ell,k)^2$. 
Therefore, we decompose the distribution $\sfp_x$ appearing in the potential as $\sfp_x = \frac34 \sfp_w + \frac14 \sfp_y$, where the distribution $\sfp_w = \frac23 \sfp + \frac13 \sfp_\cov$ is supported on vectors with Euclidean length at most $1$. 

After choosing the sign $\chi_t$ for $v$, the discrepancy vector $d_t$ becomes $d + \chi_t v$. For ease of notation, define $s_{k}(x) = \sinh(\lambda \cdot d^\top\Pi_k x)$ and $c_{k}(x) = \cosh(\lambda \cdot d^\top\Pi_k x)$ for any $x \in \BR^n$. Now \eqref{eqn:taylor_exp} implies that $\Delta\Phi := \Delta\Phi_1 + \Delta\Phi_2$ where
\begin{align*}
    \ \Delta\Phi_1 &\le \chi_t \cdot \frac34\left (\sum_{k\in [\kappa]} \lambda ~\BE_{w}\left[ s_{k}(w)  v^\top\Pi_k w\right]\right) + \frac34\sum_{k\in [\kappa]}\lambda^2 ~\BE_{w}\left[c_{k}(w) \cdot  w^\top\Pi_kvv^\top\Pi_kw\right] : = \chi_t L_1 + Q_1, \text{ and },\\
    \ \ \Delta\Phi_2 &\le \chi_t \cdot  \frac14\left (\sum_{k\in [\kappa]} \lambda ~\BE_{y}\left[ s_{k}(y)  v^\top\Pi_k y\right]\right) + \frac14\sum_{k\in [\kappa]} \lambda^2 ~\BE_{y}\left[c_{k}(y) \cdot e^{\lambda |v^\top \Pi_ky|} y^\top\Pi_kvv^\top\Pi_ky\right] : = \chi_t L_2 + Q_2.
\end{align*}

Since our algorithm chooses sign $\chi_t$ to minimize the potential increase, taking expectation over the incoming vector $v$, we get
\begin{align*}
    \ \BE_{v}[\Delta\Phi] &\le -\BE_{v}[|L_1 + L_2|] + \BE_{v}[Q_1 + Q_2].
\end{align*}

We will prove the following upper bounds on the quadratic  terms (in $\lambda$) $Q_1$ and $Q_2$.
\begin{claim}\label{claim:quadratic-ban_1}
$ \BE_{v}[Q_1+Q_2] \le C\cdot \lambda^2 \sum_{k\in [\kappa]} 2^{-k} ~\BE_{x}[c_{k}(x)\|x\|_2^2]$ for an absolute constant $C>0$.
\end{claim}

On the other hand, we will show that the linear (in $\lambda$) terms $L_1 + L_2$ is also large in expectation.
\begin{claim}\label{claim:linear-ban_1} 
  $\BE_{v}[|L_1 + L_2|] \ge  \lambda B^{-1} \sum_{k\in [\kappa]} 2^{-k}~~\BE_{x} [c_{k}(x) \|x\|_2^2] - O(1)$ for some $B \le 4{\kappa} \log(\Phi^2 n \kappa)$.
\end{claim}

By our assumption of $\Phi \le T^5$, so it follows that $2\lambda \le B^{-1}$. Therefore, combining the above two claims,
$$\BE_{v}[\Delta \Phi] ~~\le~~ (2\lambda^2-\lambda B^{-1}) \left(\sum_{k\in [\kappa]} 2^{-k} ~ \BE_{x}\left[c_{k}(x)\|x\|_2^2\right]\right) + C ~~\le~~ C ,$$
which finishes the proof of  \lref{lemma:gen-drift-ban} assuming the claims.
\end{proof}
\vspace*{8pt}

To prove the missing claims, we need the following property that follows from the sub-exponential tail of the input distribution $\sfp$.

\begin{lemma}\label{lemma:subexp_1}
There exists a constant $C>0$, such that for every integer $k\in [\kappa]$, and any $y \in \im(\Pi_k)$ satisfying $\|y\|_2 \le \frac14\sqrt{\min\{2^k, n\}}$, the following holds 
\[\BE_{v \sim \sfp}\left[e^{\lambda |v^\top  y|} \cdot |v^\top y|^2\right] \le C\cdot 2^{-k}\cdot\|y\|_2^2 \text { for all } \lambda \le 1.\] 
\end{lemma}

We remark that this is the only step in the proof which requires the sub-exponential tail, as otherwise the exponential term above may be quite large. It may however be possible to exploit some more structure from the test vectors $y$ and the discrepancy vector to prove the above lemma without any sub-exponential tail requirements from the input distribution. 


\begin{proof}
    As $y \in \im(\Pi_k)$, we have that $v^\top y = v^\top \Pi_k y$ which is a scalar sub-exponential random variable with zero mean and variance at most 
    \[\sigma_y^2 ~~:=~~ \BE_v[|v^\top \Pi_k y|^2] ~~\le~~ \|\Pi_k \cov \Pi_k\|_{\op}\|y\|_2^2 ~~\le~~ 2^{-k}\|y\|_2^2 ~~\le~~ 1/16.\] 
    
   Using Cauchy-Schwarz and  \pref{prop:logconcave}, we get that
    \begin{align*}
        \BE_v\left[e^{\lambda |v^\top  y|} \cdot |v^\top y|^2\right] ~~\le~~ \sqrt{\BE_v\left[e^{2\lambda |v^\top  y|}\right]} \cdot \sqrt{ \BE_v\left[|v^\top y|^4\right]} ~~\le~~ C \cdot \BE_v\left[|v^\top \Pi_k y|^2\right] ~~\le~~ C \cdot 2^{-k}~\|y\|_2^2,
    \end{align*}
    where the exponential term is bounded since $\sigma_y \le 1/4$.
\end{proof}
\vspace*{8pt}

\begin{proof}[Proof of  \clmref{claim:quadratic-ban_1}]
Recall that $\BE_{v}[vv^\top] = \cov$ which satisfies $\Pi_k \cov \Pi_k = 2^{-k} \Pi_k $. Therefore, using linearity of expectation,
\begin{align}\label{eqn:int1_1}
    \ \BE_{v}[Q_1] ~=~ \frac34 \sum_{k\in [\kappa]} \lambda^2 ~ \BE_{w}[c_{k}(w) \cdot w^\top\Pi_k \cov \Pi_k w] ~&=~ \lambda^2 \cdot \frac34 \sum_{k\in [\kappa]}  2^{-k} ~ \BE_{w}[c_{k}(w) \cdot w^\top \Pi_k w] \notag\\
    \ &\le~ 2\lambda^2 \cdot \frac34  \sum_{k\in [\kappa]} 2^{-k} ~ \BE_{w}[c_{k}(w)\|w\|_2^2].
\end{align}

We next use \lref{lemma:subexp_1} to  bound the second quadratic term 
\[
\BE_v[Q_2] ~=~ \frac14\sum_{k\in [\kappa]} \lambda^2 ~\BE_{y}\left[c_{k}(y) \cdot e^{\lambda |v^\top \Pi_ky|} y^\top\Pi_kvv^\top\Pi_ky\right] .\] 
For any $k\in [\kappa]$ and any $y \in \im(\Pi_k)$ that is in the support of $\sfp_y$, we have that 
\[\lambda \|\Pi_k y\|_2 \le \lambda\cdot \|y\|_2 ~\le~  \lambda / \smin(k) ~\le~ \lambda \cdot \frac{1}{10\lambda} \cdot \sqrt{\dim(H_k) }  ~\le~ \frac14\sqrt{\min\{n,2^k\}}.\]
On the other hand, if $y \in \im(\Pi_{k'})$ for $k'\neq k$, then the above quantity is zero.
\lref{lemma:subexp_1} then implies that for any $y$ in the support of $\sfp_y$,
\[ \BE_{v}[e^{|\lambda v^\top \Pi_k y|} \cdot |\lambda v^\top\Pi_k y|^2] ~\le~ C_1 \cdot 2^{-k} \|\lambda \Pi_k y\|_2^2 ~\le~ C_1 \lambda^2 \cdot 2^{-k} \|y\|_2^2,\]
where $C_1$ is some absolute constant. 
Therefore, we obtain the following bound
\begin{align}\label{eqn:int2_1}
    \ \BE_{v}[Q_2] ~\le~ C_1 \cdot \lambda^2 \cdot\sum_{k \in [\kappa]}  2^{-k}~\BE_{y}[c_{k}(y) \|y\|_2^2] .
\end{align}

Summing up \eqref{eqn:int1_1} and \eqref{eqn:int2_1} finishes the proof of the claim. 
\end{proof} 
\vspace*{8pt}

\begin{proof}[Proof of  \clmref{claim:linear-ban_1}]
Let $L = L_1 + L_2$. To lower bound the linear term, we proceed similarly as in the proof of \clmref{claim:linear} and use the fact that $|L(v)| \ge {\|f\|^{-1}_{\infty}} \cdot f(v) \cdot L(v)$ for any real-valued non-zero function $f$. We will choose the function $f(v) = d^\top\Pi v \cdot \ind_{\CG}(v)$ where $\CG$ will be the event that $|d^\top\Pi v|$ is small which we know is true because of \lref{lemma:gen-tail-ban}.


In particular, set $\delta^{-1} = \lambda^{-2} n \cdot\Phi \cdot \log(4n \Phi)$ and let $\CG$ denote the set of vectors $v$ in the support of $\sfp$ such that $\lambda|d^\top \Pi v| \le {\kappa} \cdot \log (4 \Phi/\delta) := B$. 
Then, $f(v) = d^\top\Pi v \cdot \ind_{\CG}(v)$ satisfies $\|f\|_\infty \le \lambda^{-1} B$, and we can lower bound,
\begin{align}\label{eqn:lterm_1}
    \ \BE_{v}[|L|] &\ge \frac{\lambda}{\lambda^{-1} B} \cdot \frac34 \sum_{k \in [\kappa]} \BE_{vw} [s_{k}(w) \cdot d^\top \Pi v \cdot v^\top \Pi_k w\cdot \ind_{\CG}(v)] \notag\\
    \ & ~~~~~~~ + \frac{\lambda}{\lambda^{-1}B} \cdot \frac14\sum_{k \in [\kappa]} \BE_{vy} [s_{k}(y) \cdot d^\top \Pi v \cdot v^\top \Pi_k y\cdot \ind_{\CG}(v)] \notag \\
    \ &= \frac{\lambda^2}{B} \cdot \frac34 \sum_{k \in [\kappa]} \BE_{w} [s_{k}(w) \cdot d^\top \Pi \cov \Pi_k w] ~-~ \frac{\lambda^2}{B} \cdot \frac34 \sum_{k \in [\kappa]} \BE_{w} [s_{k}(w) \cdot d^\top \Pi \cov_\err \Pi_k w] \notag \\
    \ &\qquad + \frac{\lambda^2}{B}  \cdot \frac14 \sum_{k \in [\kappa]} \BE_{y} [s_{k}(y) \cdot d^\top \Pi \cov \Pi_k y] ~-~\frac{\lambda^2}{B}\cdot \frac14 \sum_{k \in [\kappa]} \BE_{y} [s_{k}(y) \cdot d^\top \Pi \cov_\err \Pi_k y],
\end{align}
where $\cov_\err = \BE_{v}[vv^\top (1-\ind_{\CG}(v))]$ satisfies $\|\cov_\err\|_{\op} \le \BP_{v \sim \sfp}(v \notin \CG) \le \delta$ using  \lref{lemma:tail}.

To bound the terms involving $\cov$ in \eqref{eqn:lterm_1}, we recall that $s_k(x) = \sinh(\lambda d^\top \Pi_k x)$ and $c_k(x) = \cosh(\lambda d^\top \Pi_k x)$. Using $\Pi \cov \Pi_k = 2^{-k} \Pi_k$ and the fact that $\sinh(a)a \ge \cosh(a)|a| - 2$ for any $a \in \BR$, we have  
\begin{align*}
    \lambda ~\BE_{w} [s_{k}(w) \cdot d^\top \Pi \cov \Pi_k w] ~=~ 2^{-k} ~\BE_{w} [s_{k}(w) \cdot \lambda d^\top \Pi_k w] ~\ge~ 2^{-k}~ \left(\BE_{w} [c_{k}(w)|\lambda d^\top\Pi_k w|] - 2\right) ,
\end{align*}
and similarly for $y$.

The terms with $\cov_\err$ can be upper bounded using $\| \cov_\err \|_\op \leq \delta$. 
In particular, we have  
\[|d^\top \Pi \cov_\err \Pi_k x | \le \|\Pi d\|_2 \|\cov_\err\|_{\op}\|x\|_2 \le \delta \|\Pi d\|_2 \|x\|_2.\]

Since $\Pi = \sum_{k\in[\kappa]} \Pi_k$ and $(\Pi_k)_{k \in [\kappa]}$ are orthogonal projectors, \lref{lemma:chaining} implies that  $\|\Pi d\|_2 \le \lambda^{-1}  \log(4n\Phi) \sqrt{ n}$. 
Moreover, we have $\|w\|_2 \leq 1$ and $\|y\|_2 \le \min_k \{1/\smin(k)\} \le \frac1{10\lambda}\cdot \sqrt{n}$. Then, by our choice of $\delta^{-1} = \lambda^{-2} n\Phi \cdot \log(4n \Phi)$, we have 
\[
\lambda ~|d^\top \Pi \cov_\err \Pi_k x| ~~\le~~ \delta \lambda^{-1} n \log(4n \Phi) ~~=~~ \Phi^{-1}.
\]

Plugging the above bounds in \eqref{eqn:lterm_1}, we obtain 
\begin{align}\label{eq:int1_1}
    \BE_{v}[|L|] &~~\ge~~  \frac{\lambda}{B} \cdot \frac34 \sum_{k \in [\kappa]} 2^{-k} ~\BE_{w} [c_{k}(w)|\lambda  d^\top\Pi_k w|] +  \frac{\lambda}{B} \cdot \frac14 \sum_{k \in [\kappa]}  2^{-k} ~ \BE_{y} [c_{k}(y)|\lambda  d^\top\Pi_k y|] - 4
\end{align}
where we used the upper bound $\sum_{k \in [\kappa]} \BE_x[|s_k(x)|] \le \Phi$ to control the error term involving $\cov_{\err}$.

To finish the proof, we bound the two terms in \eqref{eq:int1_1} separately. We first use the inequality that $\cosh(a)a \ge \cosh(a) - 2$ for all $a \in \BR$ and the fact that $\|w\|_2 \le 1$ for every $w$ in the support of $\sfp_w$ to get that 
\begin{equation}\label{eq:int2_1}
    \BE_{w} [c_{k}(w)|\lambda  d^\top\Pi_k w|] ~~\ge~~ \BE_{w} [c_{k}(w)] - 2 ~~\ge~~ \BE_{w} [c_{k}(w)\|w\|_2^2] - 2.
\end{equation}

To bound the second term in \eqref{eq:int1_1}, we recall that the entire probability mass assigned to length $r$ vectors (i.e. $\epsilon(\ell,k) = 1/r$) in the support of $\sfp_y$ is at most $2^{-2r^2}$, where $r \ge 1/4$. Let $\CE$ be the event that $|\lambda  d^\top\Pi_k y| \le \|y\|_2^2$. 
Note that $c_k(y) \|y\|_2^2 \le 2^{r^2}r^2$ if $\|y\|_2=r$. This implies that 
\begin{align}\label{eq:int3_1}
    \BE_{y} [c_{k}(y)|\lambda  d^\top\Pi_k y|] &~\ge~ \BE_{y} [c_{k}(y)\|y\|_2^2] - \BE_{y} [c_{k}(y) \|y\|_2^2 \cdot \ind_\CE(y)] \notag \\
    \ &~\ge~ \BE_{y} [c_{k}(y)\|y\|_2^2] -  \int_{1/4}^{\infty} 2^{-2r^2} 2^{r^2}r^2 ~\ge~ \BE_{y} [c_{k}(y)\|y\|_2^2] - 1.
\end{align}

Since $\sfp_x = \frac34 \sfp_w + \frac14 \sfp_y$, plugging \eqref{eq:int2_1} and \eqref{eq:int3_1} into \eqref{eq:int1_1} give that 
 $\BE_{v}[|L|] \ge  \lambda B^{-1} \sum_{k \in [\kappa]} 2^{-k}~~\BE_{x} [c_{k}(x) \|x\|_2^2] - C$, for some constant $C>0$,
which completes the proof of the claim. 
\end{proof}

%% file: multicolor.tex

\renewcommand{\CT}{\mathcal{T}}

\section{Generalization to Weighted Multi-Color Discrepancy}
\label{sec:multicolor}

In this section, we prove \thmref{thm:multicolor-intro} which we restate below for convenience. 

\multicolor*

\thmref{thm:multicolor-intro} follows from a black-box way of converting an algorithm for the signed discrepancy setting to the multi-color setting. 

In particular, for a parameter $0\le \lambda \le 1$, let $\Phi: \BR^n \to \BR_+$ be a potential function satisfying 
\begin{equation}\label{eqn:potential}
    \begin{aligned}
    \ & \Phi(d+\alpha v) \le \Phi(d) + \lambda\alpha L_d(v) + \lambda^2\alpha^2 Q_d(v) ~~~~~~~\text{ for every } d,v \in \BR^n \text{ and }  |\alpha|\le 1, \text{ and}, \\
    \ & -\lambda \cdot\BE_{v \sim \sfp}[|L_d(v)|] +  \lambda^2\cdot \BE_{v \sim \sfp}[Q_d(v)] = O(1) ~~~\text{ for any } d \text{ such that } \Phi(d) \le 3T^5,
\end{aligned}
\end{equation}
where $L_d: \BR^n \to \BR$ and $Q_d: \BR^n \to \BR_+$ are arbitrary functions of $v$ that depend on $d$. 

One can verify that the first condition is always satisfied for the potential functions used for proving \thmref{thm:gen-disc} and \thmref{thm:gen-disc-ban}, while the second condition holds for $\lambda = O(1/\log^{2}(nT))$ because of \lref{lemma:drift-komlos} and \lref{lemma:gen-drift-ban}. 

Moreover, for parameters $n$ and $T$, let $B_{\arbnorm{\cdot}}$ be such that if the potential  $\Phi(d) = \Phi$, then the corresponding norm $\arbnorm{d} \le B_{\arbnorm{\cdot}} \log(nT\Phi)$. Part (b) of \lref{lemma:tail} implies that for any test set $\SE$ of $\poly(nT)$ vectors contained in the unit Euclidean ball, if the norm $\|\cdot\|_* = \max_{z \in \SE} |\ip{\cdot}{z}|$, then $B_{\arbnorm{\cdot}} =  O(\log^3(nT))$. Similarly, if $\arbnorm{\cdot}$ is given by a symmetric convex body with Gaussian measure at least $1/2$, then \lref{lemma:chaining} implies that $B_{\arbnorm{\cdot}}= O(\log^4(nT))$. 

We will use the above properties of the potential $\Phi$ to give a greedy algorithm for the multi-color discrepancy setting.



\subsection{Weighted Binary Tree Embedding}

We first show how to embed the weighted multi-color discrepancy problem into a binary tree $\mathcal{T}$ of height $O(\log(R\eta)$. 
For each color $c$, we create $\lfloor w_c \rfloor $ nodes with weight $w_c/\lfloor w_c \rfloor \in [1,2]$ each.
The total number of nodes is thus $M_{\ell} = \sum_{c \in [R]} \lfloor w_c \rfloor = O(R \eta)$. 
In the following, we place these nodes as the leaves of an (incomplete) binary tree.

Take the height $h = O(\log (R \eta))$ to be  the smallest exponent of $2$ such that  $2^h \geq M_{\ell}$.
We first remove $2^h - M_{\ell} < 2^{h-1}$ leaves from the complete binary tree of height $h$ such that none of the removed leaves are siblings. 
Denote the set of remaining leaves as $\SL(\mathcal{T})$.
Then from left to right, assign the leaves in $\SL(\mathcal{T})$ to the $R$ colors so that leaves corresponding to the same color are consecutive. 
For each leaf node $\ell \in \SL(\mathcal{T})$ that is assigned the color $c \in [R]$, we assign it the weight $w_\ell = w_c/\lfloor w_c \rfloor$. 

We index the internal nodes of the tree as follows: for integers $0 \le j \le h-1$ and $0\le k \le 2^{j}$, we use $(j,k)$ to denote the $2^k$-th node at depth $j$. Note that the left and right children of a node $(j,k)$ are the nodes $(j+1,2k)$ and $(j+1,2k+1)$. The weight $w_{j,k}$ of an internal node $(j,k)$ is defined to be sum of weights of all the leaves in the sub-tree rooted at $(j,k)$. This way of embedding satisfies certain desirable properties which we give in the following lemma.

\begin{lemma}[Balanced tree embedding] \label{lem:balanced_tree}
For the weighted (incomplete) binary tree $\mathcal{T}$ defined above, for any two nodes $(j,k)$ and $(j,k')$ in the same level,
\begin{align*}
    1/4 \leq w_{j,k}/ w_{j,k'} \leq 4.
\end{align*} 
\end{lemma}

\begin{proof}
Observe that each leaf node $\ell \in \SL(\mathcal{T})$ has weight $w_\ell \in [1,2]$. Moreover, for each internal node $(h-1,k)$ in the level just above the leaves, at least one of its children is not removed in the construction of $\mathcal{T}$. Therefore, it follows that $w_{j,k} = a_{j,k} 2^{h-j}$ for some $a_{j,k} \in [1/2,2]$ and similarly for $(j,k')$. The lemma now immediately follows from these observations. 
\end{proof}

\paragraph{Induced random walk on the weighted tree.} 
Randomly choosing a leaf with probability proportional to its weight induces a natural random walk on the tree $\mathcal{T}$: the walk starts from the root and moves down the tree until it reaches one of the leaves. Conditioned on the event that the walk is at some node $(j,k)$ in the $j$-th level, it goes to left child $(j+1, 2k)$ with probability $\qlt = w_{j+1,2k}/w_{j,k}$ and to the right child $(j+1,2k+1)$ with probability $\qrt = w_{j+1,2k+1}/w_{j,k}$. Note that by Lemma~\ref{lem:balanced_tree} above, we have that both $\qlt, \qrt \in [1 / 20 , 19/20]$ for each internal node $(j,k)$ in the tree. Note that $w_{j,k}/w_{0,0}$ denotes the probability that the random walk passes through the vertex $j,k$.


\subsection{Algorithm and Analysis}

Recall that each leaf $\ell \in \SL(\CT)$ of the tree $\CT$ is associated with a color. Our online algorithm will assign each arriving vector $v_t$ to one of the leaves $\ell \in \SL(\CT)$ and its color will then be the color of the corresponding leaf.

For a leaf $\ell \in \SL(\CT)$, let $d_\ell(t)$ denote the sum of all the input vectors that are associated with the leaf $\ell$ at time $t$. For an internal node $(j,k)$, we define $d_{j,k}(t)$ to be the sum  $\sum_{\ell \in \SL(\CT_{j,k})} d_\ell(t)$ where $\SL(
\CT_{j,k})$ is the set of all the leaves in the sub-tree rooted at $(j,k)$. Also, let $\dlt(t) = d_{j+1,2k}(t)$ and $\drt(t) = d_{j+1,2k+1}(t)$ be the vectors associated with the left and right child of the node $(j,k)$. 

Finally let, 
$$d^-_{j,k}(t) =  \frac{\dlt(t)/\qlt - \drt(t)/\qrt}{1/\qlt + 1/\qrt} = \qrt \dlt(t) - \qlt\drt,$$ 
denote the weighted difference between the two children vectors for the $(j,k)$-th node of the tree. 

\paragraph{Algorithm.} For $\beta = 1/(400h)$, consider the following potential function
\begin{align*}
\Psi_t = \sum_{j,k \in \CT} \Phi(\beta~ d^-_{j,k}(t)),
\end{align*}
where the sum is over all the internal nodes $(j,k)$ of $\CT$.

The algorithm assigns the incoming vector $v_t$ to the leaf $\ell \in \SL(\CT)$, so that the increase in the potential $\Psi_t - \Psi_{t-1}$ is minimized. The color assigned to the vector $v_t$ is then the color of the corresponding leaf $\ell$.

We show that if the potential $\Phi$ satisfies \eqref{eqn:potential}, then the drift for the potential $\Psi$ can be bounded.
\begin{lemma} \label{clm:driftmulti}
If at any time $t$, if $\Psi_{t-1} \le T^5$, then the following holds
$$ \BE_{v_t \sim \sfp}[\Delta\Psi_t] := \BE_{v_t \sim \sfp}[\Psi_t - \Psi_{t-1}] = O(1).$$ 
\end{lemma}

Using standard arguments as used in the proof of \thmref{thm:gen-disc}, this implies that with high probability $\Psi_t \le T^5$ at all times $t$.

Moreover, the above potential also gives a bound on the discrepancy because of the following lemma.

\begin{lemma} \label{lem:node_prop_to_weights}
If $\Psi_t \le T^5$, then
$\disc_t = O(\beta^{-1} h \cdot B_{\arbnorm{\cdot}} \cdot \log (nT \Psi_t)) = O(h^2 \cdot B_{\arbnorm{\cdot}} \cdot \log (nT))$.
\end{lemma}

Combined with part (b) of \lref{lemma:tail} and \lref{lemma:chaining}, the above implies \thmref{thm:multicolor-intro}. Next we prove \lref{lem:node_prop_to_weights} and \lref{clm:driftmulti} in that order. 

\subsubsection*{Bounded Potential Implies Low Discrepancy}

For notational simplicity, we fix a time $t$ and drop the time index below. 

\begin{proof}[Proof of \lref{lem:node_prop_to_weights}]
First note that $\Phi(\beta \cdot d^-_{j,k}) \le \Psi$, and therefore, $\arbnorm{d^-_{j,k}} \le \beta^{-1} B(\arbnorm{\cdot}) := U$ for every internal node $(j,k)$. 

We next claim by induction that the above implies the following for every internal node $(j,k)$,
\begin{align} \label{eq:disc_from_root}
\arbnorm{d_{j,k} - d_{0,0} \cdot \frac{w_{j,k}}{w_{0,0}}} \leq  \beta_j U,
\end{align}
where $\beta_j = 1+19/20+\cdots+(19/20)^j$.

The claim is trivially true for the root. For an arbitrary node  $(j+1,2k)$ at depth $j$ that is the left child of some node $(j,k)$, we have that 
\begin{align*}
    \arbnorm{{d_{j+1,2k}} - d_{0,0}\cdot \frac{w_{j+1,2k}}{w_{0,0}}} &\le  \arbnorm{{d_{j+1,2k}} - d_{j,k}\cdot \frac{w_{j+1,2k}}{w_{j,k}}} + \qlt \cdot \arbnorm{d_{j,k} - d_{0,0} \cdot \frac{w_{j,k}}{w_{0,0}}} \\
    \ & \le \arbnorm{{\dlt} - d_{j,k}\cdot \qlt} +  \qlt \beta_j U,
\end{align*}
since $w_{j+1,2k}/w_{j,k}=\qlt$ and $\qlt,\qrt \in[1/20,19/20]$. Note that $d_{j,k} = \dlt + \drt$, so the first term above equals $\arbnorm{d^-_{j,k}}$. Therefore, it follows that $\arbnorm{{d_{j+1,2k}} - d_{0,0}\cdot ({w_{j+1,2k}}/{w_{0,0}})} \le \beta_{j+1} U$. The claim follows analogously for all nodes that are the right children of its parent.


To see the statement of the lemma,  consider any color $c \in [R]$. We say that an internal node has color $c$ if all its leaves are assigned color $c$. A maximal color-$c$ node is a node that has color $c$ but its ancestor doesn't have color $c$. We denote the set of maximal $c$-color node to be $\SM_c$.  
Notice that $|\SM_c| \leq 2 h$ since $c$-color leaves are consecutive. Also, note that $\sum_{(j,k) \in \SM_c} w_{j,k} = w_c$ and that $\sum_{(j,k) \in \SM_c} d_{j,k} = d_c$ is exactly the sum of vectors with color $c$.  
Therefore, we have 
\begin{align*}
\arbnorm{d_c/w_c - d_{0,0}/w_{0,0}} \le   \arbnorm{d_c - d_{0,0} \cdot \frac{w_c}{w_{0,0}}} \le  \sum_{(j,k) \in \mathcal{M}_c} \arbnorm{ d_{j,k} - d_{0,0} \cdot\frac{w_{j,k} }{w_{0,0}} } = O(h \cdot U), 
\end{align*}
where the first inequality follows since $w_c \ge 1$ and the last follows from \eqref{eq:disc_from_root}.

Thus, for any two colors $c \neq c'$, we have
\begin{align*}
\disc_t(c,c') = \arbnorm{ \frac{d_c / w_c - d_{c'} / w_{c'}}{1/ w_c + 1/w_{c'}}} 
 \le \arbnorm{ \frac{d_c / w_c - d_{0,0} / w_{0,0}}{1/ w_c + 1/w_{c'}}} + \arbnorm{ \frac{d_{c'} / w_{c'} - d_{0,0} / w_{0,0}}{1/ w_c + 1/w_{c'}}} =  O(h \cdot U).
\end{align*}
This finishes the proof of the lemma. 
\end{proof}

\subsubsection*{Bounding the Drift}



\begin{proof}[Proof of \clmref{clm:driftmulti}]
We fix the time $t$ and write $d^{-}_{j,k} = d^{-}_{j,k}(t-1)$. Let $X_{j,k}(\ell) \cdot v_t$ denote the change of $d^-_{j,k}$ when the leaf chosen for $v_t$ is $\ell$. More specifically,
$X_{j,k}(\ell)$ is $\qrt$ if the leaf $\ell$ belongs to the left sub-tree of node $(j,k)$, is $-\qlt$ if it belongs to the  right sub-tree, and is $0$ otherwise. Then, $d^-_{j,k}(t) = d^-_{j,k} + X_{j,k}(\ell)\cdot v_t$ if the leaf $\ell$ is chosen.

By our assumption on the potential, we have that $\Delta\Psi_t \leq \beta \lambda L + \beta^2 \lambda^2 Q$ where 
\begin{align*}
L &= \sum_{(j,k) \in \calP(\ell)} X_{j,k}(\ell) \cdot L_{j,k}(v_t) \\
Q &= \sum_{(j,k) \in \calP(\ell)} X_{j,k}(\ell)^2 \cdot Q_{j,k}(v_t) ,
\end{align*} 
and $\calP(\ell)$ is the root-leaf path to the leaf $\ell$.

Consider choosing leaf $\ell$ (and hence the root-leaf path $\calP(\ell)$) randomly in the following way: 
First pick a uniformly random layer $j^* \in \{0,1, \cdots, h-1\}$ (i.e., level of the tree), then starting from the root randomly choose a child according to the random walk probability for all layers except $j^*$; for layer $j^*$, suppose we arrive at node $(j^*,k)$, we pick the left child if $L_{j^*,k}(v_t) \leq 0$, and the right child otherwise. Note that conditioned on a fixed value of $j_*$, this ensures that $\BE_\ell[X_{j,k} L_{j,k}(v_t)]$ is always negative if $j=j_*$ and is zero otherwise.

Since we follow the random walk before layer $j^*$,  for a fixed choice of $j^*$ we get a node in its layer proportional to their weights. Let us write $\SN_j$ for the set of all nodes at depth $j$. In expectation over the randomness of the input vector $v_t$ and our random choice of leaf $\ell$, we have
\begin{align*}
\E_{v_t,\ell}[L] 
&\leq - \frac{1}{h} \cdot \sum_{j=0}^{h-1} \sum_{k \in \SN_j} \frac{w_{j,k}}{\sum_{j \in \SN_j} w_{j,k}} \cdot \min\{\qlt,\qrt\} \cdot \E_{v_t}[|L_{j,k}|]. 
\end{align*}

For the $Q$ term, recall that one is randomly picking a child until layer $j^*$, in which one picks a child depending on $L_{j^*,k}$, and then we continue randomly for the remaining layers. Note that since $Q$ is always positive, this can be at most $20$ times a process that always picks a random root-leaf path, since we have $\qlt, \qrt \in [1/20,19/20]$. Therefore, we have
\begin{align*}
\E_{v_t, \ell}[Q] \leq 20 \cdot \sum_{j=0}^{h-1} \sum_{k \in \SN_j} \frac{w_{j,k}}{\sum_{j \in \SN_j} w_{j,k}} \cdot \E_{v_t}[Q_{j,k}].
\end{align*}

By our choice of $\beta=1/(400h)$, the above implies that
\begin{align*}
\E_{v_t}[\Delta\Psi_t] &\leq - \sum_{j=0}^{h-1} \sum_{k \in \SN_j} \frac{w_{j,k}}{\sum_{j \in \SN_j} w_{j,k}} \cdot \left(-\frac{\beta \lambda }{20h}\E_{v_t}[|L_{j,k}|] +  20\beta^2\lambda^2\E_{v_t}[Q_{j,k}]\right)\\
\ & \leq - \sum_{j=0}^{h-1} \sum_{k \in \SN_j} \frac{w_{j,k}}{\sum_{j \in \SN_j} w_{j,k}} \cdot \frac{1}{8000h^2}\cdot \left(-\lambda\E_{v_t}[|L_{j,k}|] +  \lambda^2\E_{v_t}[Q_{j,k}]\right) = O(1).
\end{align*} 
Since the algorithm is greedy, the leaf $\ell$ it assigns to the incoming vector $v$ produces an even smaller drift, so this completes the proof.
\end{proof}